\documentclass[10pt,reqno]{article}
\usepackage{amsmath,amssymb,amsthm}
\usepackage{esint}
\usepackage{bm}
\usepackage{url}
\usepackage{subfigure}
\usepackage{graphicx,color}
\usepackage{algorithm}
\usepackage{algpseudocode}
\usepackage{algorithmicx}

\usepackage{booktabs,multirow}
\usepackage{hhline}
\usepackage{setspace}

\usepackage{caption}
\usepackage{hyphsubst}
\usepackage{caption}

\usepackage{float}
\usepackage{filecontents}
\usepackage{hyperref}

\hypersetup{colorlinks,citecolor=blue,linkcolor=blue, urlcolor=blue}
\usepackage{pdflscape}

\def\Be{{\bf e}}

\def\Bn{{\bf n}}

\def\Bw{{\bf w}}
\def\Bx{{\bf x}}

\DeclareMathAlphabet{\itbf}{OML}{cmm}{b}{it}

\def\by{{{\bf y}}}
\def\bx{{{\bf x}}}
\def\hbx{{\hat{\bx}}}

\def\be{{{\itbf e}}}
\def\bn{{{\bf n}}}
\def\bu{{{\bf u}}}
\def\bv{{{\bf v}}}
\def\bw{{{\bf w}}}
\def\bd{{{\bf d}}}

\newcommand{\bG}{\mathbf{G}}
\newcommand{\bGamma}{\mathbf{\Gamma}}

\def\bgamma{{\boldsymbol{\gamma}}}
\def\bpsi{{\boldsymbol{\psi}}}

\def\bphi{{\boldsymbol{\varphi}}}

\newcommand{\bp}{\mathbf{p}}

 \newcommand{\hbd}{\hat{\bd}}
\newcommand{\RR}{\mathbb{R}}
\newcommand{\ZZ}{\mathbb{Z}}
\newcommand{\R}{\mathbb{R}}

\newcommand{\bS}{\mathbf{S}}

\newcommand{\K}{{\kappa}}
\newcommand{\dis}{\displaystyle}
\newcommand{\OL}{\mathcal{L}}
\newcommand{\bW}{\mathbf{W}}

\newcommand{\I}{\mathbf{I}}

\newcommand{\C}{\mathcal{C}}

\newcommand{\CC}{\mathbb{C}}

\newcommand{\bH}{\mathbf{H}}

\newcommand{\ds}{\displaystyle}
\newcommand{\NN}{\mathbb{N}}

\newcommand{\Kcal}{\mathcal{K}}

\def\nm{\noalign{\medskip}}
\newcommand{\bV}{\mathbf{V}}

\newcommand{\bP}{\mathbf{P}}
\newcommand{\bJ}{\mathbf{J}}
\newcommand{\bF}{\mathbf{F}}

\newcommand{\snr}{{\rm SNR}}

\newcommand{\noise}{{\rm noise}}

\newtheorem{thm}{Theorem}[section]
\newtheorem{cor}[thm]{Corollary}
\newtheorem{lem}[thm]{Lemma}

\newtheorem{defn}[thm]{Definition}
\newtheorem{rem}[thm]{Remark}

\DeclareMathOperator*{\argmin}{arg\,min}

\numberwithin{equation}{section}

\newcommand{\pathfigures}{Figures/}
\graphicspath{{\pathfigures}}

\begin{document}

\title{
Two-Dimensional  Elastic Scattering Coefficients and Enhancement of Nearly Elastic Cloaking
\thanks{This research was supported by the Ministry of Science, ICT and Future Planning through the National Research Foundation of Korea grant NRF-2015H1D3A106240 (to A.W. through the Korea Research Fellowship Program), by the National Research Foundation of Korea under Grants NRF-2016R1A2B3008104  and NRF-2014R1A2A1A11052491 (A.W. and J.C.Y), and R\&D Convergence Program of National Research Council of Science and Technology of Korea grant no. CAP-13-3-KERI (to A.W. and J.C.Y.). The work of G.H. is partially supported by the NSFC grant no. 11671028 and the 1000-Talent Program of Young Scientists in China. }
}
\author{
Tasawar Abbas
\thanks{\footnotesize Department of Mathematics and Statistics, Faculty of Basic and Applied Sciences, International Islamic University, 44000, Islamabad, Pakistan (tasawar44@hotmail.com).}
\and
Habib Ammari
\thanks{\footnotesize Department of Mathematics, ETH Z\"{u}rich, Ramistrasse 101, CH-8092 Z\"{u}rich, Switzerland (habib.ammari@math.ethz.ch).}
\and
Guanghui Hu
\thanks{\footnotesize  Beijing Computational Science Research Center,
Building 9, East Zone,  ZPark II, No.10 Xibeiwang East Road, Haidian District, 100193 Beijing, China (hu@csrc.ac.cn).}
\and
Abdul Wahab
\thanks{\footnotesize Bio Imaging and Signal Processing Laboratory, Department of Bio and Brain Engineering, Korea Advanced Institute of Science and Technology, 291 Daehak-ro, Yuseong-gu, Daejeon 305-701, Korea (wahab@kaist.ac.kr, jong.ye@kaist.ac.kr).}\,
\thanks{Address all correspondence to A. Wahab at Email: wahab@kaist.ac.kr, Ph. +82-42-350-4360, Fax: +82-350-4310.}
\and
Jong Chul Ye \footnotemark[5]
}
\maketitle
\begin{abstract}
The concept of scattering coefficients has played a pivotal role in a broad range of inverse scattering and imaging problems in acoustic, and electromagnetic media. In view of their promising applications in inverse problems related to mathematical imaging and elastic cloaking, the notion of elastic scattering coefficients  of an inclusion is presented  from the perspective of boundary layer potentials and a few properties are discussed.  A reconstruction algorithm is developed and analyzed for extracting the elastic scattering coefficients  from multi-static response measurements of the scattered field in order to cater to inverse scattering problems. The decay rate, stability and error analyses, and the estimate of maximal resolving order in terms of the signal-to-noise ratio are discussed. Moreover, scattering-coefficients-vanishing structures are designed and their utility for enhancement of nearly elastic cloaking is elucidated.
\end{abstract}

\noindent {\footnotesize {\bf AMS subject classifications 2000.} Primary: 35L05, 35R30, 74B05, 74J20, 78A46.}

\noindent {\footnotesize {\bf Key words.} Elastic scattering, Scattering coefficients, Elastic cloaking, Inverse scattering.}

\section{Introduction}\label{s:intro}

The matrix theory for scattering of waves by obstacles of various geometric nature (e.g., crack, inclusions, cavities, etc.) has been the subject of numerous investigations since 1960's. The transformation matrix (or so-called T-matrix) approach has been of particular interest. It was introduced by Waterman \cite{Waterman-em} in 1965  in order to study the scattering of electromagnetic waves and later on to discuss the scattering of acoustic \cite{Waterman-a} and elastic waves \cite{Waterman}. The method is based on the fact that incident and scattered fields  both admit multipolar expansions  in terms of wave functions thanks to the Jacobi-Anger decomposition and  wave addition theorems. The coefficients of the expansion of the field scattered by a given obstacle are connected to those of the incident field by an infinite matrix, coined as T-matrix. Such a matrix is independent of the choice of incident field, and dependents only on the morphology of the obstacles and the frequency of incidence. Therefore, the same matrix can be used to connect any incident field to the corresponding scattered field for a given obstacle and frequency of incidence. Since its inception, this concept has received great attention by researchers for various applications of wave scattering. It is especially the case when numerical computations become necessary because it is computationally efficient to use T-matrices for numerical simulations of scattering phenomena \cite{Ganesh10, Martin03}. In addition, this approach has been of great advantage in dealing with multiple scattering since the T-matrices corresponding to different obstacles can be combind easily using translation-addition theorems. The interested readers are referred to the monographs \cite{Martin06, Varadan80, Dassios2000} and to the survey articles \cite{Rev1, Rev2, Rev3} for detailed accounts of related work.

In T-matrix approach, the series expansions of the incident and the scattered fields are obtained in relevant complete orthogonal bases. Then, the series coefficients of the scattered field are linked to those of the incident field by T-matrix using Lippmann-Schwinger representation of the scattered field based on the conditions imposed on the boundary of the obstacle \cite{Waterman-em, Waterman-a}.  Most often, especially in elasticity,  the incident field is expanded using a real basis of cylindrical or spherical wave functions (composed of real surface harmonics and Bessel functions).  This renders a symmetric T-matrix and a unitary scattering matrix  (or simply  S-matrix, defined in terms of T-matrix) thanks to the principles of  reciprocity and the conservation of energy in loss-less media \cite{Varath, Waterman}.

In this article, we deal with the elastic scattering by an inclusion using a  complex basis of eigenvectors of the Lam\'{e} equation  (defined in terms of complex vector harmonics) for incident field and a rigorous  integral representation of the scattered field in terms of layer potentials, unlike the standard T-matrix approach  presented in \cite{Varath, Waterman}.  The elements of the resulting matrix connecting the coefficients of the scattered and incident fields  are coined as elastic scattering coefficients (ESC).  The impetus behind this study is the enhancement of nearly elastic cloaking and promising applications of the ESC in mathematical imaging and inverse scattering.

The ESC can be perceived as a natural extension of the concept of elastic moment tensor \cite{Princeton}  with respect to frequency dependence.  Thanks to the integral representation of the scattered field in terms of layer potentials, the ESC can be explicitly defined using boundary densities solving a system of boundary integral equations.  In addition, these frequency-dependent geometric objects contain rich information about the contrast of material parameters, high-order shape oscillations, frequency profile, and the maximal resolving power of an imaging setup.  The concept of scattering coefficients in acoustic and electromagnetic media has been effectively used for inverse medium scattering and mathematical understanding of super-resolution phenomena in imaging
\cite{medium, Shape}.

The scattering coefficients were proved to be felicitous to design enhanced near invisibility cloaks  in acoustic and electromagnetic media \cite{Helmholtz, Maxwell}.
The invisibility cloaking is an exciting area of interest nowadays. The idea of invisibility cloaking has been proved to be scientifically realizable in many investigations (see, for instance,  \cite{PenSchSmi, GLU, GLU2, Greenleaf1, Greenleaf2, Leo, MBW}). Significant progress has been made recently on the control of conductivity equations \cite{Ammari1, GLU, GLU2},  acoustic \cite{Helmholtz, BaoLiu, ChenChan},  electromagnetic \cite{Maxwell, BaoLiuZou} and elastic waves \cite{HL2015, Diatta, Diatta2, Farhat} using curvilinear transformations of coordinates. In fact, an invisibility cloak is perceived as a meta-material that maps a concealment region into a surrounding shell by virtue of a transformation and thereby making the material parameters strongly heterogeneous and anisotropic, however fulfilling impedance matching with the surrounding vacuum.

The aim of this article is three-fold:  (a) to present the ESC and discuss some of their properties  indispensable for this investigation, (b) to propose and analyze a reconstruction framework for the recovery of the ESC to cater to direct and inverse scattering problems,  (c) to design scattering coefficients vanishing elastic structures for the enhancement of the nearly elastic cloaking devices.  We restrict ourselves to a two-dimensional case, however, the three-dimensional case is amenable to the same treatment with slight changes. First, we consider elastic wave scattering from an inclusion embedded in an otherwise homogeneous medium and define the associated ESC using the cylindrical eigenfunctions of the Lam\'e equation and the integral representation of the scattered field in terms of layer potentials. Then, a least-squares optimization algorithm is designed for the reconstruction of the significant ESC from the full aperture Multi-Static Response (MSR) data collected using a circular acquisition system. Multistatic imaging involves two steps. The first step consists in recording the waves  generated by point sources on an  array of receivers. The second step consists in processing the recorded matrix data in order to estimate some features of the medium  \cite{iakovleva, MSRI-Book}.  The stability, truncation error and maximal resolving order of the reconstruction procedure are analytically quantified.
 The results contained in this article can cater to many inverse scattering problems, especially for shape identification and classification in elastic media. The interested readers are referred to \cite{LimYu} and articles cited therein for comprehensive details on shape identification in elastic media.
Finally, we design mathematical structures with vanishing scattering coefficients (S-vanishing structures) and elaborate a framework for the enhancement of nearly elastic cloaking.

The contents of this article are organized in the following manner. Some notation and a few preliminary results on layer potential theory of elastic scattering  are collected in Section \ref{s:form}. In Section \ref{s:ESC},  the ESC  are defined and their important features are discussed.   Section \ref{s:recon} is dedicated to the reconstruction framework for the ESC. The enhancement procedure for elastic cloaking is elaborated in Section \ref{s:cloaking}. Finally, in  Section \ref{s:conc}, we sum up the important contributions of this investigation and discuss interesting applications of the ESC in mathematical imaging.

\section{Elements of Layer Potential Theory}\label{s:form}

Since this article is concerned with elastic scattering and the integral formulation of the scattered field is the key component to define the ESC, we feel it best to pause and introduce some background material from layer potential theory for linear time-harmonic elasticity. For details beyond those we provide in this section, please refer to the monographs \cite{Princeton, Kupradze79}.

\subsection{Preliminaries and Notation}\label{ss:pre}

To simplify matters, we confine ourselves to the two-dimensional case throughout this article. However, the three-dimensional case is amenable to the same treatment with appropriate changes.

For any sufficiently smooth, open and bounded domain $\Omega\subset\RR^2$ with $\C^2-$ boundary $\partial \Omega$, we define $L^2(\Omega)$ in the usual way with norm
$$
\|u\|_{L^2(\Omega)}:=\left(\int_\Omega|u|^2d\bx\right)^{1/2},
$$
and the Hilbert space $H^1(\Omega)$ by
$$
H^1(\Omega):=\big\{u\in L^2(\Omega)\,|\, \nabla u\in L^2(\Omega)\big\},
$$
 with norm
$$
\|u\|_{H^1(\Omega)}:=\left(\|u\|_{L^2(\Omega)}^2+\|\nabla u\|_{L^2(\Omega)}^2\right)^{1/2}.
$$
We define $H^2(\Omega)$ as the space of functions $u\in H^1(\Omega)$ such that $\partial_{ij} u\in L^2(\Omega)$  for all $i,j=1,2$, and  $H^{3/2}(\Omega)$ as the interpolation space $[H^1(\Omega), H^2(\Omega)]_{1/2}$. Let $\mathbf{t}$ be the tangent vector to $\partial\Omega$ at point $\bx$ and let $\partial/\partial {\mathbf{t}}$ denote the tangential derivative. Then, we say that $u\in H^1(\partial\Omega)$ if $u\in L^2(\partial \Omega)$ and $\partial u/\partial \mathbf{t}\in L^2(\partial \Omega)$. Refer to \cite{Bergh} for further details.

Consider a homogeneous isotropic elastic material, occupying a bounded domain $D\subset\R^2$ with connected $\C^2-$boundary $\partial D$,  compressional and shear moduli $\lambda_1\in\RR_+$ and $\mu_1\in\RR_+$ respectively, and density $\rho_1\in\RR_+$. Let the exterior domain $\RR^2\setminus \overline{D}$ be loaded with different elastic material having parameters $\rho_0,\lambda_0,\mu_0\in\RR_+$ such that
\begin{equation}\label{Lame-Conditions}
(\lambda_0-\lambda_1)^2+(\mu_0-\mu_1)^2 \neq 0
\quad\text{and}\quad
(\lambda_0-\lambda_1)(\mu_0-\mu_1)\geq 0.
\end{equation}
To facilitate latter analysis,  we introduce piecewise defined parameters
\begin{align*}
\lambda(\bx)&:=\lambda_0\chi_{(\RR^2\setminus D)}(\bx)+\lambda_1\chi_{(D)}(\bx),
\\\nm
\mu(\bx)&:=\mu_0\chi_{(\RR^2\setminus D)}(\bx)+\mu_1\chi_{(D)}(\bx),
\\\nm
\rho(\bx)&:=\rho_0\chi_{(\RR^2\setminus D)}(\bx)+\rho_1\chi_{(D)}(\bx),
\end{align*}
where $\chi_\Omega$ represents the characteristic function of a domain $\Omega$.
We also define the linear elasticity operator  $\OL_{\lambda_0,\mu_0}$  and the surface traction operator (or conormal derivative) ${\partial}/{\partial \nu}$, associated with parameters $(\lambda_0,\mu_0)$ by
\begin{equation*} 
\OL_{\lambda_0,\mu_0} [\bw]:=\left[\mu_0\Delta\bw+(\lambda_0+\mu_0)\nabla\nabla\cdot\bw\right],
\end{equation*}
and
\begin{equation*} 
\dis\frac{\partial \bw}{\partial \nu}:=\lambda_0(\nabla\cdot\bw)\bn+2\mu_0\left(\nabla^s\bw\right)
\bn,
\end{equation*}
for all sufficiently smooth vector fields $\bw:\RR^2\to \RR^2$, where $\bn\in\RR^2$ represents the outward unit normal to $\partial D$, $\nabla^s\bw=(\nabla\bw+(\nabla\bw)^\top)/2$ is the linear elastic strain  and the superscript $\top$ reflects a transpose operation.

Let $\omega>0$ be the angular frequency of the mechanical oscillations. We denote the outgoing fundamental solution of the time-harmonic elasticity equation in $\RR^2$ with parameters $(\lambda_0,\mu_0,\rho_0)$ by $\bGamma^\omega$, i.e., for all $\bx\in\RR^2$,
\begin{equation*}  
(\OL_{\lambda_0,\mu_0}+\rho_0\omega^2\mathcal{I})\bGamma^\omega(\bx)=-\delta_{0}(\bx)\I_2,\qquad\forall \bx\in\RR^2,
\end{equation*}
subject to the \emph{Kupradze's} outgoing radiation conditions. Here $\delta_\by$ is the Dirac mass at $\by$, $\mathcal{I}$ is the identity operator, and $\I_2\in\RR^{2\times 2}$ is the identity matrix.  Let $\K_{\alpha}:={\omega}/{c_{\alpha}}$ for $\alpha=P,S$, where the constants $c_{S}=\sqrt{{\mu_0}/{\rho_0}}$ and $c_{P}=\sqrt{{\lambda_0+ 2 \mu_0}/{\rho_0}}$ refer to background shear and pressure wave speeds respectively.
It is well known that (see, for instance, \cite{Morse}) 
\begin{equation}\label{Green_fun}
\bGamma^\omega(\bx)=\frac{1}{\mu_0}\left[\dis\left(\I_2+\frac{1}{\K_S^2}\nabla\nabla^\top\right)g(\bx,\K_S)-\frac{1}{\K_S^2}\nabla\nabla^\top g(\bx,\K_P)\right],\quad \bx\in\R^2\setminus\{ \mathbf{0}\}.
\end{equation}
The function $g(\cdot,\K_\alpha)$ is the fundamental solution to the Helmholtz operator $-(\Delta+\K_\alpha^2\mathcal{I})$  in $\RR^2$ with wave-number $\K_\alpha\in \RR_+$ ($\K_\alpha=\K_P$ or $\K_\alpha$), i.e.,
$$
(\Delta+\K_\alpha^2\mathcal{I}) g(\bx,\K_\alpha)=-\delta_{\bf 0}(\bx),\quad \bx\in\RR^2,
$$
subject to the \emph{Sommerfeld's} outgoing radiation condition
$$
\lim_{|\bx|\to+\infty}|\bx|^{1/2}\left[\frac{\partial g(\bx,\K_\alpha)}{\partial\bn}-i\K_\alpha g(\bx,\K_\alpha)\right]=0, \qquad \bx\in\RR^2,
$$
where ${\partial}/{\partial\bn}$ represents the normal derivative. In two dimensions,
\begin{equation}
g(\bx,\K_\alpha)=\dis\frac{i}{4} H^{(1)}_0(\K_\alpha|\bx|),\quad\forall\bx\in\RR^2\setminus\{\mathbf{0}\},\label{green-fn-g}
\end{equation}
where $H^{(1)}_0$ is the Hankel function of first kind of order zero (see, for instance, \cite{nedelec}). Throughout this article, we use the convention $\bGamma^\omega(\bx,\by)=\bGamma^\omega(\bx-\by)$, i.e.,
\begin{equation*}
(\OL_{\lambda_0,\mu_0}+\rho_0\omega^2\mathcal{I})\bGamma^\omega(\bx,\by)=-\delta_{\by}(\bx)\I_2,\qquad\forall \bx,\by\in\RR^2.
\end{equation*}
We also reserve the notation $\alpha$ and $\beta$ to represent pressure (P) and shear (S) wave-modes, i.e., $\alpha, \beta \in\{P,S\}$.

\subsection{Scattered Field and Integral Representation}\label{ss:int}

Let us begin this subsection by introducing the elastic single layer potential
\begin{eqnarray*}
\mathcal{S}_D^\omega[\bphi](\bx) :=\dis\int_{\partial D}\bGamma^\omega(\bx,\by)\bphi(\by)d\sigma(\by),\qquad \qquad \bx\in\RR^2\setminus\partial D,
\end{eqnarray*}
for all densities $\bphi\in L^2(\partial D )^2$.
Here and throughout this article $d\sigma$ denotes the infinitesimal boundary differential element. We also need the boundary integral operator
$$
(\Kcal_D^\omega)^*[\bphi](\bx) = {\rm p.v.}\dis\int_{\partial D}\frac{\partial}{\partial\nu_\bx}\bGamma^\omega(\bx,\by)\bphi(\by)d\sigma(\by),\quad{\rm a.e.}\quad\bx\in\partial D,
$$
for all $\bphi\in L^2(\partial D)^2$. Here  p.v.  stands for Cauchy principle value of the integral and the surface traction of matrix  $\bGamma^\omega$ is defined column-wise, i.e., for all constant vectors $\mathbf{p}\in\RR^2$
$$
\left[\frac{\partial\bGamma^\omega}{\partial\nu}\right]\mathbf{p}=\frac{\partial\left[\bGamma^\omega\mathbf{p}\right]}{\partial\nu}.
$$
We recall that the traces $ \mathcal{S}^\omega_D[\bphi]\big|_{\pm}$ and ${\partial(\mathcal{S}_D^\omega[\bphi])}/{\partial \nu}\big|_{\pm}$ are well-defined and satisfy the jump conditions (see, for instance, \cite{Dahlberg})
\begin{eqnarray}
\label{Sjumps}
\begin{cases}
\dis\mathcal{S}_D^\omega[\bphi]\big|_{+}(\bx)=\mathcal{S}_D^\omega[\bphi]\big|_{-}(\bx),
\\\nm
\ds\frac{\partial(\mathcal{S}_D^\omega[\bphi])}{\partial \nu}\Big|_{\pm}(\bx)
=\left(\pm\frac{1}{2}I+(\mathcal{K}_D^\omega)^*\right)\bphi(\bx), \quad{\rm a.e.}\quad\bx\in\partial D.
\end{cases}
\end{eqnarray}
Here and throughout this investigation subscripts $+$ and $-$ indicate the limiting values across the boundary $\partial D$ from outside and from inside domain $D$ respectively, i.e., for any function  $\psi$
$$
\left(\psi(\bx)\right)\big|_\pm=\lim_{\epsilon\to 0^+} \psi(\bx\pm \epsilon\mathbf{n}), \quad \bx\in\partial D.
$$

Consider a time harmonic incident elastic field $\bu^{\rm inc}$ satisfying
\begin{equation}\label{U}
(\OL_{\lambda_0,\mu_0}+\rho_0\omega^2\mathcal{I})\bu^{\rm inc}(\bx) =0, \quad \forall \bx\in\RR^2.
\end{equation}
Then, the total displacement field $\bu^{\rm tot}=\bu^{\rm sc}+\bu^{\rm inc}$ in the presence of inclusion $D$ satisfies the transmission problem
\begin{equation}
\begin{cases}
(\OL_{\lambda,\mu}+\rho\omega^2\mathcal{I})\bu^{\rm tot}(\bx)=0, \quad  \forall\bx\in\RR^2,
\\
\nm
\text{
$(\bu^{\rm tot}-\bu^{\rm inc})(\bx)=: \bu^{\rm sc}(\bx)$  satisfies  Kupradze's radiation condition as $|\bx|\to+\infty$,}\label{u}
\end{cases}
\end{equation}
where $\bu^{\rm sc}$ denotes the scattered field. It is recalled  that the total field $\bu^{\rm tot}$ admits the integral representation (see\cite[Theorem 1.8]{Princeton}) 
\begin{equation}\label{u-int-rep}
\bu^{\rm tot}(\bx,\omega)=
\begin{cases}
\bu^{\rm inc}(\bx,\omega)+\mathcal{S}^\omega_D[\bpsi](\bx,\omega), & \bx\in\RR^2\setminus\overline{D},
\\\nm
\widetilde{\mathcal{S}}^\omega_D[\bphi](\bx,\omega), & \bx\in D,
\end{cases}
\end{equation}
in terms of single layer potentials $\mathcal{S}^\omega_D$ and $\widetilde{\mathcal{S}}^\omega_D$, where the densities $\bphi, \bpsi\in L^2(\partial D)^2$ satisfy the system of integral equations
\begin{equation}\label{integral-system}
\begin{pmatrix}
\widetilde{\mathcal{S}}_D^\omega & -{\mathcal{S}}_D^\omega
\\\nm
\ds
\frac{\partial}{\partial\widetilde{\nu}}\widetilde{\mathcal{S}}_D^\omega\Big|_-
&
\ds
-\frac{\partial}{\partial\nu}{\mathcal{S}}_D^\omega\Big|_+
\end{pmatrix}
\begin{pmatrix}
\bphi
\\\nm
\bpsi
\end{pmatrix}
=
\dis
\begin{pmatrix}
\bu^{\rm inc}
\\\nm
\ds\frac{\partial\bu^{\rm inc}}{\partial\nu}
\end{pmatrix}\Bigg|_{\partial D}.
\end{equation}
Here a superposed $\sim$ is used to distinguish the  single layer potential and the surface traction  defined using interior Lam\'e parameters $(\lambda_1,\mu_1,\rho_1)$. To simply matters, the dependence of $\bu^{\rm inc}$, $\bu^{\rm sc}$, $\bu^{\rm tot}$, $\bphi$ and $\bpsi$ on frequency $\omega$ is suppressed unless it is necessary.

The following result from \cite[Theorem 1.7]{Princeton} guarantees the unique solvability of the system \eqref{integral-system} and consequently that of problems \eqref{u} and \eqref{u-int-rep}.
\begin{thm}\label{thmSolve}
Let $D$ be a Lipschitz bounded domain in $\RR^2$ with parameters $0<\lambda_1,\mu_1,\rho_1<\infty$ satisfying condition \eqref{Lame-Conditions} and let $\omega^2\rho_1$ be different from Dirichlet eigenvalues of the operator $-\OL_{\lambda_1,\mu_1}$ on $D$. For any function $\bu^{\rm inc}\in H^1(\partial D)^2$ there exists a unique solution $(\bphi,\bpsi)\in L^2(\partial D)^2\times L^2(\partial D)^2$ to the integral system \eqref{integral-system}.
Moreover, there exists a constant $C>0$ such that
\begin{align}
\label{stability}
\|\bphi\|_{L^2(\partial D)^2}+\|\bpsi\|_{L^2(\partial D)^2}
\leq C \left(\|\bu^{\rm inc}\|_{H^1(\partial D)^2}+\left\|\frac{\partial\bu^{\rm inc}}{\partial \nu}\right\|_{L^2(\partial D)^2}\right).
\end{align}
\end{thm}

\section{Elastic Scattering Coefficients}\label{s:ESC}

This section is dedicated to defining ESC in two dimensions. To facilitate the definition of ESC, we first recall some background material on cylindrical eigenfunctions of the Lam\'e equation, and present the multipolar expansions of the exterior scattered elastic field and the Kupradze fundamental solution $\bGamma^\omega$ in the next subsection.

\subsection{Cylindrical Elastic Waves and Multipolar Expansions}\label{ss:wavefun}

In the sequel, $\hbx:=\bx/|\bx|$ for all $\bx\in\RR^2\setminus\{\mathbf{0}\}$ and $\mathbb{S}:=\{\bx\in\RR^2\,|\, \bx\cdot\bx=1\}$. Any position vector $\bx:=(x_1,x_2)\in\RR^2$  is equivalently expressed as $\bx=(|\bx|\cos\varphi_\bx,|\bx|\sin\varphi_\bx)$, where $\varphi_\bx\in[0,2\pi)$ denotes the polar angle of $\bx$.
Denote by  $\{\hat{\mathbf{e}}_r, \hat{\mathbf{e}}_{\theta}\}$ the orthonormal basis vectors for the polar coordinate system in two dimensions, i.e.,
$$
\hat{\mathbf{e}}_r=(\cos\varphi_\bx, \sin \varphi_\bx)^\top 
\quad\text{and}\quad
\hat{\mathbf{e}}_{\theta}=(-\sin\varphi_\bx, \cos \varphi_\bx)^\top.
$$
We will also require the following \emph{curl} operators in $\RR^2$
$$
\vec{\nabla}_\perp\times f := \left(\partial_2 f , -\partial_1 f\right)^\top
\text{and}
\quad
\nabla_\perp\times \bw:= \partial_1 w_2-\partial_2 w_1,
$$
for any smooth scalar function $f$ and vector $\bw:=(w_1,w_2)^\top$.

 Consider the complex surface vector harmonics in two-dimensions defined by
\begin{align*}
\mathbf{P}_m(\hbx):=e^{im\varphi_\bx}\hat{\mathbf{e}}_r
\quad\text{and}\quad
\mathbf{S}_m(\hbx):=e^{im\varphi_\bx}\hat{\mathbf{e}}_{\theta}, \qquad\qquad\forall\,m\in\mathbb{Z}.
\end{align*}
It is known, see \cite{Morse} for instance, that these cylindrical surface vector potentials enjoy the  orthogonality properties
\begin{align}
&\int_{\mathbb{S}} \bP_n(\hbx)\cdot\overline{\bP_m(\hbx)} d\sigma(\hbx) =2\pi\delta_{nm},
\label{PP}
\\\nm
&\int_{\mathbb{S}} \bS_n(\hbx)\cdot\overline{\bS_m(\hbx)} d\sigma(\hbx) =2\pi\delta_{nm},
\label{SS}
\\\nm
&\int_{\mathbb{S}} \bP_m(\hbx)\cdot\overline{\bS_m(\hbx)} d\sigma(\hbx) =0,
\qquad\forall\,n,m\in\mathbb{Z},
\label{SP}
\end{align}
where $\delta_{nm}$ is the Kronecker's delta function and the superposed bar indicates complex conjugation.

Let $H^{(1)}_m$ and $J_m$ be the cylindrical Hankel and Bessel functions of first kind of order $m\in\mathbb{Z}$, respectively. For each $\K>0$, $\K\in\{\K_P, \K_S\}$, the wave functions $v_{m}(\cdot,\K)$ and $w_{m}(\cdot,\K)$ are constructed by
\begin{align*}
v_{m}(\bx,\K):=H^{(1)}_m(\K|\bx|)e^{im\varphi_{\bx}}
\quad\text{and}\quad
w_{m}(\bx,\K):=J_m(\K|\bx|)e^{im\varphi_{\bx}}.
\end{align*}
It is easy to verify that $v_{m}$ is an outgoing radiating solution  to the Helmholtz equation
$$
\Delta v+\K^2 v=0, \qquad\text{in }\, \RR^2\setminus\{\mathbf{0}\},
$$
and that $w_m$ is an entire solution  to
$$
\Delta v+\K^2 v=0, \qquad\text{in }\, \RR^2.
$$
Using surface vector harmonics $\bP_{m}$, $\bS_{m}$ and functions $v_{m}$, $w_{m}$, we define
\begin{align}
\bH^{P}_{m}(\bx,\K_P)
:=&
\nabla v_m(\bx,\K_P)
=
\K_P\left(H^{(1)}_m(\K_P|\bx|)\right)'\bP_m(\hbx)+\frac{im}{|\bx|}H^{(1)}_m(\K_P|\bx|)\bS_m(\hbx),\label{UP}
\\
\bH^{S}_{m}(\bx,\K_S)
:=&
\vec{\nabla}_\perp\times\left(v_m(\bx,\K_S)\right)
=
\frac{im}{|\bx|}H^{(1)}_m(\K_S|\bx|)\bP_m(\hbx)-\K_S\left(H^{(1)}_m(\K_S|\bx|)\right)'\bS_m(\hbx),\label{US}
\end{align}
and
\begin{align}
\bJ^{P}_{m}(\bx,\K_P)
:=&
\nabla w_m(\bx,\K_P)
=
\K_P\left(J_m(\K_P|\bx|)\right)'\bP_m(\hbx)+\frac{im}{|\bx|}J_m(\K_P|\bx|)\bS_m(\hbx),\label{UPtilde}
\\
\bJ^{S}_{m}(\bx,\K_S)
:=&
\vec{\nabla}_\perp\times\left( w_m(\bx,\K_S)\right)
=
\frac{im}{|\bx|}J_m(\K_S|\bx|)\bP_m(\hbx)-\K_S\left(J_m(\K_S|\bx|)\right)'\bS_m(\hbx),\label{UStilde}
\end{align}
for all  $\K_\alpha>0$ and $m\in\mathbb{Z}$. Here
\begin{align*}
\left(H^{(1)}_m\right)'(t):=\frac{d}{dt}\left[H^{(1)}_m(t)\right]
\quad\text{and}\quad
\left(J_m\right)'(t):=\frac{d}{dt}\left[J_m(t)\right].
\end{align*}
For simplicity, the dependence of $\bJ^\alpha_m$ and $\bH^\alpha_m$ on wave-numbers $\K_\alpha$ is suppressed henceforth.

The functions $\bJ^{P}_{m}$ and $\bJ^{S}_{m}$ are the longitudinal and transverse interior eigenvectors of the Lam\'e system in $\RR^2$. Similarly, ${\bH}^{P}_{m}$ and ${\bH}^{S}_{m}$ are the exterior eigenvectors of the Lam\'e system in $\RR^2\setminus\{\mathbf{0}\}$. Following result on the completeness and linear independence of the interior eigenvectors $\left(\bJ^P_m, \bJ^S_m\right)$ and exterior eigenvectors $\left({\bH}^P_m, {\bH}^S_m\right)$  with respect to $L^2(\partial D)^2-$norm holds. The interested readers are referred to \cite[Lemmas 1-3]{Sevroglou} for further details.

\begin{lem}\label{lemCompleteness}
If $D\subset\RR^2$ is a bounded simply connected domain containing origin and $\partial D$ is a closed Lyapunov curve  then the set $\{\bH^P_m,\bH^S_m:\,\, m\in\mathbb{Z}\}$ is complete and linearly independent in $L^2(\partial D)^2$. Moreover, if $\rho_1\omega^2$ is not a Dirichlet eigenvalue of the Lam\'e equation  on $D$ then the set $\{\bJ^P_m,\bJ^S_m:\,\,m\in\mathbb{Z}\}$ is also complete and linearly independent in $L^2(\partial D)^2$.
\end{lem}

As a direct consequence of Lemma \ref{lemCompleteness}, corresponding to every incident field $\bu^{\rm inc}$ satisfying \eqref{U}, there exist constants $a_m^P,a_m^S\in\CC$, for all $m\in\mathbb{Z}$, such that
\begin{eqnarray}
\bu^{\rm inc}(\bx)=\sum_{m\in\mathbb{Z}}\left(a_m^S\bJ_m^S(\bx)+a_m^P\bJ_m^P(\bx)\right),\quad \bx\in\RR^2.\label{U-multipole}
\end{eqnarray}
In particular,  a general plane incident wave of the form
\begin{align}
\bu^{\rm inc}(\bx) =& \phantom{-}
\frac{1}{\rho_0c_S^2}e^{i\K_S\bx\cdot\bd}\,\bd^\perp+\frac{1}{\rho_0c_P^2}e^{i\K_P\bx\cdot\bd}\,\bd
\nonumber
\\
\nm
=&
-\left(\frac{i}{\rho_0c_S^2\K_S}\left[\vec{\nabla}_\perp\times\,e^{i\K_S\bx\cdot\bd}\right]+\frac{i}{\rho_0c_P^2\K_P}\left[\nabla  e^{i\K_P\bx\cdot\bd}\right]\right),\label{U-plane}
\end{align}
can be written in the form \eqref{U-multipole} with
\begin{eqnarray}
\label{U-multipole-Co}
a_m^\beta:=a_m^\beta(\bu^{\rm inc})=-\frac{i}{\rho_0c^2_\beta\K_\beta
}e^{im(\pi/2-\theta_\bd)}, \quad \beta\in\{P,S\},
\end{eqnarray}
where $\bd=(\cos \theta_\bd, \sin\theta_\bd)\in\mathbb{S}$ is the direction of incidence and $\bd^\perp$ is a vector perpendicular to $\bd$. In fact, this decomposition is a simple consequence of Jacobi-Anger decomposition of the scalar plane wave, i.e., 
$$
e^{i\K\bx\cdot\bd}=\sum_{m\in\mathbb{Z}}e^{im(\pi/2-\theta_\bd)}J_m(\K|\bx|)e^{im\varphi_\bx}.
$$
Moreover, for all $\bx,\by\in\RR^2$ such that $|\bx|>|\by|$ and for any vector $\bp\in\RR^2$ independent of $\bx$,
\begin{align}
\bGamma^\omega(\bx,\by)\bp=
\frac{i}{4\rho_0 \omega^2}\sum_{n\in\mathbb{Z}}
\left(\bH^P_n(\bx)
\left[\overline{\bJ^P_n(\by)}\cdot\bp\right]
+\bH^S_n(\bx)\left[\overline{\bJ^S_n(\by)}\cdot\bp\right]
\right). \label{G-multipole}
\end{align}
The interested readers are referred to Appendix \ref{AppendA} for the derivation of this expansion. Here $\bGamma^\omega(\bx,\by)\bp$ is a general elastic field generated by a point source at $\by$ in the direction $\bp$ and observed at point $\bx$ away from $\by$.

\subsection{Scattering Coefficients of Elastic Inclusions}\label{ss:scat-cof}

The multipolar expansion \eqref{G-multipole} enables us  to derive such an expansion for  the single layer potential $\mathcal{S}_D^\omega[\bpsi]$ for all $\bx\in\RR^2\setminus \overline{D}$ sufficiently far from the boundary $\partial D$. Consequently, by virtue of the integral representation \eqref{u-int-rep}, the scattered field can be expanded as
\begin{align}
\bu^{\rm sc}(\bx)
=
-\frac{i}{4\rho_0\omega^2}\sum_{n\in\ZZ}\Bigg(\bH^P_{n}(\bx)\int_{\partial D}
&
\left[\overline{\bJ^P_{n}(\by)}\cdot\bpsi(\by)\right]d\sigma(\by)
\nonumber
\\
&
+\bH^{S}_{n}(\bx)\int_{\partial D}\left[\overline{\bJ^{S}_{n}(\by)}\cdot\bpsi(\by)\right]d\sigma(\by)\Bigg).
\label{u-U-exp}
\end{align}
Let us define the elastic scattering coefficients as follows.
\begin{defn}
\label{Def}
Let $\left(\bphi^{\beta}_{m},\bpsi^{\beta}_{m}\right)$, for $m\in\ZZ$,  be the solution of \eqref{integral-system} corresponding to $\bu^{\rm inc}=\bJ^{\beta}_{m}$. Then,  the elastic scattering coefficients, $W^{\alpha,\beta}_{m,n}$, of
$D\Subset\RR^2$ are defined by
\begin{align}
W^{\alpha,\beta}_{m,n}=W^{\alpha,\beta}_{m,n}[D,\lambda_0,\lambda_1,\mu_0,\mu_1,\rho_0,\rho_1,\omega]:= \int_{\partial D} \left[\overline{\bJ^{\alpha}_{n}(\by)}\cdot\bpsi^{\beta}_{m}(\by)\right]d\sigma(\by),\quad \forall\,m,n\in\ZZ,\label{ESCoeff}
\end{align}
where $\alpha$ and $\beta$ indicate wave-modes $P$ or $S$.
\end{defn}

Following result on the decay rate of the ESC holds.
\begin{lem}\label{LemWbound}
There exist constants $C_{\alpha,\beta}>0$  for each wave-mode $\alpha,\beta=P,S$ such that
\begin{align}
\Big|W^{\alpha,\beta}_{m,n}[D,\lambda_0,\lambda_1,\mu_0,\mu_1,\rho_0,\rho_1,\omega]\Big|\leq \frac{C_{\alpha,\beta}^{|n|+|m|-2}}{|n|^{|n|-1}|m|^{|m|-1}}, \label{Wbound}
\end{align}
for all $m,n\in\ZZ$ and $|m|, |n|\to+\infty$.
\end{lem}
\begin{proof}
The proof of the estimate \eqref{Wbound} is very similar to Lemma $2.1$ in \cite{Helmholtz} for the acoustic scattering coefficients. For the sake of completeness, we fix the ideas of the proof in the sequel. Recall the asymptotic behavior
$$
J_m(t)\simeq \sqrt{\frac{1}{2\pi |m|}}\left(\frac{et}{2|m|}\right)^{|m|},
$$
of  Bessel functions of first kind  with respect to the order $|m|\to+\infty$ when the argument $t$ is fixed \cite[Formula 10.19.1]{NIST}. Then, by using the recurrence formulae (see, for instance, \cite[Formula 10.6.2]{NIST})
$$
J_m'(t)=-J_{m+1}(t)+\frac{m}{t}J_m(t)
\quad\text{and}\quad
J_m'(t)= J_{m-1}(t)-\frac{m}{t}J_m(t),
$$
one obtains
$$
|J_m'(t)|\leq \sqrt{\frac{1}{2\pi(|m|+1)}} \left(\frac{et}{2(|m|+1)}\right)^{|m|+1}+\frac{|m|}{t}\sqrt{\frac{1}{2\pi |m|}}\left(\frac{et}{2|m|}\right)^{|m|}.
$$
Consequently, by the definitions \eqref{UPtilde}-\eqref{UStilde} of $\bJ^\alpha_{n}(\bx)$  and Theorem \ref{thmSolve}, we have
$$
\left\|\bJ^\alpha_{n}\right\|_{L^2(\partial D)^2}\leq \left(\frac{C_1^\alpha}{|n|}\right)^{|n|-1},
$$
and
$$
\left\|\bpsi_m^\beta\right\|_{L^2(\partial D)^2}\leq C\left(\left\|\bJ^\beta_{m}\right\|_{L^2(\partial D)^2}+\left\|\frac{\partial\bJ^\beta_{m}}{\partial\nu}\right\|_{L^2(\partial D)^2}\right)
\leq \left(\frac{C_2^\beta}{|m|}\right)^{|m|-1},
$$
for some constants $C_1^\alpha$ and $C_2^\beta$.
Finally, the result follows by substituting the estimates for the norms of $\bpsi_m^\beta$ and $\bJ^\alpha_n$ in the definition of the scattering coefficients and choosing $C_{\alpha,\beta}$  appropriately in terms of $C_1^\alpha$ and $C_2^\beta$.
\end{proof}

\subsection{Connections with Scattered Field and Far Field Amplitudes}\label{ss:far}

Consider a general incident field $\bu^{\rm inc}$ of the form \eqref{U-multipole}. By the superposition principle,  the solution $(\bphi,\bpsi)$ of \eqref{integral-system} is given by
$$
\bpsi(\bx)= \sum_{m\in\ZZ}
\Big[a^{P}_{m}\bpsi^{P}_{m}+a^{S}_{m}\bpsi^{S}_{m}\Big]
\quad\text{and}\quad
\bphi(\bx)= \sum_{m\in\ZZ}
\Big[a^{P}_{m}\bphi^{P}_{m}+a^{S}_{m}\bphi^{S}_{m}\Big].
$$
This, together with Definition \ref{Def} of the scattering coefficients and the expansion \eqref{u-U-exp}, furnishes the asymptotic expansion
\begin{align}
\bu^{\rm sc}(\bx)
=&
\sum_{n\in\ZZ}
\left[\gamma^{P}_{n}\bH^{P}_{n}(\bx)+\gamma^{S}_{n}\bH^{S}_{n}(\bx)\right],
\label{u-U-multipole}
\end{align}
of the scattered field for all $\bx\in\RR^2\setminus \overline{D}$ sufficiently far from $\partial D$. Here
\begin{align}
\label{gamma}
\gamma_{n}^{\beta}=
\ds\sum_{m\in\ZZ}\left(d_{m}^{P} W_{m,n}^{\beta,P}+d_{m}^{S} W_{m,n}^{\beta,S}\right),\quad  \beta\in\{P,S\},
\end{align}
with $d_m^\beta:= {ia^\beta_m}/{4\rho_0\omega^2}$. Notice that for the plane incident wave given by \eqref{U-plane} the coefficients $a^\beta_m$ are given by \eqref{U-multipole-Co}. In that case 
\begin{align}
d_m^\beta:= \frac{i}{4\rho_0\omega^2}a^\beta_m=\frac{1}{(2\rho_0\omega c_\beta)^2\K_\beta}e^{im(\pi/2-\theta_\bd)}.\label{d}
\end{align}

In order to substantiate the connection between ESC and far field scattering amplitudes, we recall that the cylindrical Hankel function $H^{(1)}_n$ admits the far field behavior
\begin{align}
H^{(1)}_n(\K|\bx|)&\sim \frac{e^{i\K|\bx|}}{\sqrt{|\bx|}}\sqrt{\frac{2}{\pi\K}}e^{-i\pi(n/2+1/4)}+O\left(|x|^{-3/2}\right),\label{decayH}
\\\nm
\left(H^{(1)}_n(\K|\bx|)\right)'&\sim i \frac{e^{i\K|\bx|}}{\sqrt{|\bx|}}\sqrt{\frac{2}{\pi\K}}e^{-i\pi(n/2+1/4)}+O\left(|x|^{-3/2}\right),\label{decayHPrime}
\end{align}
as $|\bx|\to+\infty$ (see, for instance, \cite[Formulae 10.17.5 and 10.17.11]{NIST}). Here, the notation $\sim$ indicates that only leading order terms are retained. Consequently, the far field behavior of the functions $\bH^P_{n}$ and $\bH^{S}_{n}$ is predicted as
\begin{align*}
\bH^{P}_{n}(\bx)\sim\frac{e^{i\K_P|\bx|}}{\sqrt{|\bx|}} A^{\infty,P}_n\bP_n(\hbx)
\quad\text{and}\quad
\bH^{S}_{n}(\bx)\sim\frac{e^{i\K_S|\bx|}}{\sqrt{|\bx|}}A^{\infty,S}_n\bS_n(\hbx), &\quad \text{as}\,\,|\bx|\to+\infty,
\end{align*}
where
\begin{eqnarray*}
A^{\infty,P}_n:=(i+1)\K_Pe^{-in\pi/2}\sqrt{\frac{1}{\pi\K_P}}
\quad\text{and}\quad
A^{\infty,S}_n:= -(i+1)\K_Se^{-in\pi/2}\sqrt{\frac{1}{\pi\K_S}}.
\end{eqnarray*}
Thus, for all $\bx\in\RR^2\setminus \overline{D}$ such that  $|\bx|\to+\infty$,  the scattered field $\bu^{\rm sc}$
in \eqref{u-U-multipole} admits the asymptotic expansion
\begin{align}
\bu^{\rm sc}(\bx)=&\frac{e^{i\K_P|\bx|}}{\sqrt{|\bx|}}\sum_{n\in\ZZ}\left[\gamma_n^P A^{\infty,P}_n\bP_n(\hbx)\right]+\frac{e^{i\K_S|\bx|}}{\sqrt{|\bx|}}\sum_{n\in\ZZ}\left[\gamma_n^S A^{\infty,S}_n\bS_n(\hbx)\right].
\label{u-U-far}
\end{align}
On the other hand, the Kupradze radiation condition guarantees the existence of two analytic functions $\bu^\infty_P:\mathbb{S}\to\CC^2$ and $\bu^\infty_S:\mathbb{S} \to\CC^2$ such that
\begin{align}
\label{FarField}
\bu^{\rm sc}(\bx)=\frac{e^{i\K_P|\bx|}}{\sqrt{|\bx|}}\bu^\infty_P(\hbx)+\frac{e^{i\K_S|\bx|}}{\sqrt{|\bx|}}\bu^\infty_S(\hbx)+ O\left(\frac{1}{|\bx|^{3/2}}\right), \quad\text{as}\,\,|\bx|\to+\infty.
\end{align}
The functions $\bu^\infty_P$ and $\bu^\infty_S$ are respectively known as the longitudinal and transverse far-field patterns or the scattering amplitudes. Comparing \eqref{u-U-far} and \eqref{FarField} the following result is readily proved, which substantiates that the  far-field scattering amplitudes admit natural expansions in terms of scattering coefficients.
\begin{thm}\label{ThmFarAmp}
If $\bu^{\rm inc}$ is the  incident field  given by \eqref{U-multipole} then the corresponding  longitudinal and transverse scattering amplitudes  are given by
\begin{align*}
\bu^\infty_P[D,\lambda_0,\lambda_1,\mu_0,\mu_1,\rho_0,\rho_1,\omega](\hbx)=&\sum_{n\in\ZZ}\gamma_n^P A^{\infty,P}_n\bP_n(\hbx),
\\
\bu^\infty_S[D,\lambda_0,\lambda_1,\mu_0,\mu_1,\rho_0,\rho_1,\omega](\hbx)=& \sum_{n\in\ZZ}\gamma_n^S A^{\infty,S}_n\bS_n(\hbx),
\end{align*}
where the coefficients $\gamma^{P}_{n}$ and $\gamma^{S}_{n}$ are defined in \eqref{gamma}.
\end{thm}

\subsection{Reciprocity and Conservation of Energy}\label{s:properties}

In this section, we discuss some properties of the elastic scattering coefficients resulting from the principles of reciprocity and conservation of energy in elastic media. These properties can be very useful as accuracy checks or as constraints for the numerical recovery procedures for the reconstruction of the elastic scattering coefficients from measured fields discussed in Section \ref{s:recon}.

\subsubsection{Reciprocity}

Reciprocity refers to the link between the far-field amplitudes of the scattered field in two reciprocal configurations: (a) when the scattered field is radiated by a source with a specific incidence direction and is received along another direction, (b)  when the positions of source and receiver are swapped and the orientation of all momenta is reversed. More precisely, if $\bu^{\infty}_{\alpha}(\hbx;\hbd,\beta)$ is the  far field amplitude of the $\alpha-$scattered field along $\hbx$ when  a $\beta-$wave is incident along $\hbd$ then the reciprocity in elastic media refers to the properties (see, for instance, \cite[Eqs. 51-53]{Varath},  \cite{Varath2, Dassios87})
\begin{align*}
&\ds\bu^{\infty}_{\alpha}(\hbx;\hbd,\alpha) =\bu^{\infty}_{\alpha}(-\hbd;-\hbx,\alpha),
\qquad\qquad\qquad\qquad\qquad \alpha=P,S,
\\\nm
&\ds\frac{1}{\sqrt{\K_S}}\bu^{\infty}_{S}(\hbx;\hbd,P)\cdot \hat{\be}_\theta =-\frac{1}{\sqrt{\K_P}}\bu^{\infty}_{P}(-\hbd;-\hbx,S)\cdot\hat{\be}_r.
\end{align*}
In particular,  these properties dictate that the infinite matrix containing the ESC of an inclusion, defined by
$$
\mathbf{W}_\infty:=
\begin{pmatrix}
\mathbf{W}^{P,P}_\infty & \mathbf{W}^{S,P}_\infty
\\
\mathbf{W}^{P,S}_\infty & \mathbf{W}^{S,S}_\infty
\end{pmatrix}
\quad\text{with}\quad
\left(\mathbf{W}^{\alpha,\beta}_\infty\right)_{mn} := W^{\alpha,\beta}_{m,n},
\qquad\qquad\forall\,m,n\in\ZZ,
$$
should be Hermitian \cite{Varath, Waterman}. In fact, we have the following result.
\begin{lem}\label{Lem:Wsym2}
The global matrix of scattering coefficients $\mathbf{W}_\infty$ is Hermitian, i.e., for all $n,m\in\ZZ$ and $\alpha,\beta\in\{P,S\}$
\begin{align*}
W^{\alpha,\beta}_{m,n}[D,\lambda_0,\lambda_1,\mu_0,\mu_1,\rho_0,\rho_1,\omega]=\overline{W^{\beta,\alpha}_{n,m}[D,\lambda_0,\lambda_1,\mu_0,\mu_1,\rho_0,\rho_1,\omega]}.
\end{align*}
\end{lem}
\begin{proof}
Refer to Appendix \ref{AppendB}.
\end{proof}

In addition to the above reciprocity induced symmetry of $\mathbf{W}_\infty$, the following result also holds.
\begin{lem}\label{Lem:Wsym1}
For all $m,n\in\ZZ$  and $\alpha,\beta\in\{P,S\}$,
\begin{eqnarray*}
{W}_{-m,-n}^{\alpha,\beta}[D,\lambda_0,\lambda_1,\mu_0,\mu_1,\rho_0,\rho_1,\omega]= (-1)^{m+n}\overline{{W}_{m,n}^{\alpha,\beta}[D,\lambda_0,\lambda_1,\mu_0,\mu_1,\rho_0,\rho_1,\omega]}.
\end{eqnarray*}
\end{lem}
\begin{proof}
Let $\bpsi^{\beta}_{-m}$ be the unique solution of the integral system \eqref{integral-system} with $\bu^{\rm inc}(\bx):=\bJ_{-m}^\beta(\bx)$. Then, by definition
$$
{W}_{-m,-n}^{\alpha,\beta}[D]:=\int_{\partial D} \left[\overline{\bJ^{\alpha}_{-n}(\by)}\cdot\bpsi^{\beta}_{-m}(\by)\right]d\sigma(\by).
$$
On the other hand, recall that the cylindrical Bessel functions possess the connection property (see, for instance, \cite[Formula 10.4.1]{NIST})
$$
J_{-m}(\bx)=(-1)^m J_{m}(\bx).
$$
Therefore, for all $m\in\ZZ$ and $\bx\in\RR^2$,
$$
\bJ^\beta_{-m}(\bx):=(-1)^{m}\overline{\bJ^\beta_{m}(\bx)}.
$$
Consequently, by uniqueness of the solution of \eqref{integral-system}, one obtains the relation $\bpsi_{-m}^\beta=(-1)^{m}\overline{\bpsi_{m}^\beta}$ and thus
\begin{align*}
W_{-m,-n}^{\alpha,\beta}[D]=&
(-1)^{m+n}\int_{\partial D} \left[\bJ^{\alpha}_{n}(\by,\K_\alpha)\cdot\overline{\bpsi^{\beta}_{m}(\by)}\right]d\sigma(\by)
=(-1)^{m+n}\overline{W_{m,n}^{\alpha,\beta}[D]},
\end{align*}
which is the required result.
\end{proof}

\subsubsection{Conservation of Energy}

The conservation of energy within any bounded elastic domain, say $\Omega$, without energy dissipation and compactly containing the inclusion  $D$  states that the rate of the energy flux across the boundary of $\Omega$  must be zero. This render the so-called  \emph{optical theorem of scattering} or \emph{forward scattering  amplitude theorem} that links the scattering cross section or the rate at which the energy is scattered by $D$ to the far field amplitude of its scattering signature.
Consider for instance that the domain $\Omega$ is a disk with  a very large radius $R\to+\infty$. Then, by virtue of the far-field expansion of the scattered elastic field, the optical theorem states that  (see, for instance, \cite{Varath2, Dassios87})
\begin{align}
\int_0^{2\pi}& \left(\frac{1}{\K_P}\left|\bu^\infty_P(\hbx;\hbd,\alpha)\right|^2+\frac{1}{\K_S}\left|\bu^\infty_S(\hbx;\hbd,\alpha)\right|^2\right)d\theta
\nonumber
\\
=&
\begin{cases}
\ds
-2\sqrt{\frac{2\pi}{\K_P}}\Im m \left\{\sqrt{i} \bu^\infty_P(\hbd;\hbd, P).\hat{\be}_r\right\}, &{\alpha=P},
\\\nm
\ds
2\sqrt{\frac{2\pi}{\K_S}}\Im m \left\{\sqrt{i} \bu^\infty_S(\hbd;\hbd, S).\hat{\be}_\theta\right\}, &{\alpha=S}.
\end{cases}
\label{ScatteringCrossSec}
\end{align}
This leads to the following statement of the optical theorem in terms of $\mathbf{W}_\infty$.
\begin{thm}\label{thmOptic}
\begin{align}
\ds\frac{1}{4\rho_0\omega^2} \mathbf{W}_\infty\overline{\mathbf{W}_\infty}=-\Im m{\mathbf{W}_\infty}=\frac{i}{2}\left(\mathbf{W}_\infty-\overline{\mathbf{W}_\infty}\right).
\label{W-constraint}
\end{align}
\end{thm}
\begin{proof}
Refer to Appendix \ref{AppendC}.
\end{proof}
It is worthwhile mentioning that the identity \eqref{W-constraint} is slightly different from the one proved by \cite[Eq. 100]{Varath} for T-matrices. This is simply  due to the choice of bases functions $\mathbf{J}^\alpha$ and $\mathbf{H}^\alpha$ and the corresponding multipolar expansion of the fundamental solution.  The matrix $\mathbf{W}_\infty$ is not unitary, however, the corresponding scattering matrix can be proved to be unitary by invoking  relation \eqref{W-constraint}. The interested readers are referred to \cite{Waterman} for related discussion.
The relation \eqref{W-constraint} provides a natural constraint on $\mathbf{W}_\infty$ which can be efficiently used to reduce the ill-posedness of the reconstruction procedure for ESC of an inclusion from scattering data (see, for instance, \cite{Ottaviani}).

\section{Reconstruction of Scattering Coefficients}\label{s:recon}

The matrix approach has been proved to be very efficient and powerful for direct and inverse scattering simulations and to study the effect of material contrast, shape and orientation of an obstacle with respect to the incident field. For acoustic and electromagnetic wave scattering different techniques have been devised to recover truncated T-matrices from the scattered wave in the far field  regime (see, for instance, \cite{Ganesh10, Ganesh, Martin03, NullField, Martin06, Rev1, Rev2} and references therein). On the other hand, up to the best of knowledge of the authors, no attempt has been made to  recover T-matrices for elastic wave scattering. Unlike T-matrix computation for forward scattering, ESC admit explicit expressions in terms of boundary densities $\bpsi_m$ which satisfy the integral system \eqref{integral-system}. In order to solve \eqref{integral-system}, there are several efficient multipole methods available at hand. Therefore, the ESC corresponding to a given inclusion can be directly computed. Our interest in the ESC lies in their applications in inverse scattering problems wherein one only has a limited information of the scattered field in the exterior domain whereas the location and the morphology of the inclusion are usually unknown. Towards this end, we design a mathematical procedure based on least-squares minimization using MSR data and provide theoretical details of the stability, truncation error, and maximal resolving order in this section. To simplify the matters,  we consider the full aperture case with a circular  acquisition system.

\subsection{MSR Data Acquisition}\label{ss:msr}
Let $\{\bx_s\}_{s=1,\cdots, N_s}$ and $\{\bx_r\}_{r=1,\cdots,N_r}$ be the sets of locations of the point sources and point receivers, and $\{\bd_s,\bd_s^\perp\}_{s=1,\cdots, N_s}$ and $\{\bd_r,\bd_r^\perp\}_{r=1,\cdots, N_r}$ (such that $\bd_s\cdot\bd_s^\perp=0=\bd_r\cdot\bd_r^\perp$) be the corresponding unit directions of incidence  and reception  respectively for some $N_r,N_s\in\NN$. Let the points $\{\bx_s\}$ and $\{\bx_r\}$  be uniformly distributed over the circle $\partial B_R(\textbf{0})$ with radius $R$ centered at origin such that $|\bx_r|=R=|\bx_s|$ and $\theta_r=\theta_{\bx_r}=2\pi r/N_r$ and $\theta_s=\theta_{\bx_s}=2\pi s/N_s$. We consider a regime in which $R$ is large enough so that the terms of order $O(R^{-3/2})$ are negligible. For simplicity, it is  assumed that $D$ contains the origin which is reasonable since we are in sufficiently far field regime and the inclusion $D$ can be envisioned as \emph{sufficiently centered} in $B_R(\textbf{0})$. This assumption can be easily removed by invoking transformation rules for the scattering coefficients corresponding to coordinate translation. 

Let $\bF_s$ and $\bG_s$, for all $s=1,\cdots,N_s$, be the pressure and shear type incident waves emitted from point $\bx_s$ with direction of incidence $\bd_s$, i.e.,
\begin{align*}
\bF_s(\bx):=\frac{1}{\rho_0 c_P^2}\bd_s e^{i\K_P\bx\cdot\bd_s}
\quad\text{and}\quad
\bG_s(\bx):=\frac{1}{\rho_0 c_S^2}\bd_s^\perp e^{i\K_S\bx\cdot\bd_s}.
\end{align*}
Let $\bu^{\rm tot}_{\bF_s}(\bx)$ and $\bu^{\rm tot}_{\bG_s}(\bx)$ (resp.  $\bu^{\rm sc}_{\bF_s}(\bx)$ and $\bu^{\rm sc}_{\bG_s}(\bx)$) be the corresponding total (resp. scattered) fields. For all incident fields $\bF_s$ and $\bG_s$  the scattered fields are recorded at all points $\bx_r$ along the directions $\bd_r$ and $\bd_r^\perp$ so that four MSR matrices $\mathbf{A}^{\ell,\ell'}= ({A}^{\ell,\ell'}_{sr})_{s=1,\cdots N_s, r= 1,\cdots,N_r}$, for $\ell,\ell'\in\{\parallel,\perp\}$, are obtained with elements
\begin{align*}
{A}^{\parallel,\parallel}_{sr}=&\Big(\left[\bu^{\rm sc}_{\bF_s}(\bx_r)\right]\cdot \bd_r\Big)_{sr},
\\\nm
{A}^{\parallel,\perp}_{sr}=&\Big(\left[\bu^{\rm sc}_{\bF_s}(\bx_r)\right]\cdot \bd_r^\perp\Big)_{sr},
\\\nm
{A}^{\perp,\parallel}_{sr}=&\Big(\left[\bu^{\rm sc}_{\bG_s}(\bx_r)\right]\cdot \bd_r\Big)_{sr},
\\\nm
{A}^{\perp,\perp}_{sr}=&\Big(\left[\bu^{\rm sc} _{\bG_s}(\bx_r)\right]\cdot \bd_r^\perp\Big)_{sr},
\end{align*}
at a given frequency $\omega$. Note that, by virtue of expansions \eqref{U-multipole}, \eqref{u-U-multipole} and \eqref{gamma}, the elements of the MSR matrices admit the  expansions
\begin{align}
A^{\parallel,\parallel}_{sr}=\sum_{n,m\in\ZZ}
&d_{m}^{P}(s)\left(W_{m,n}^{P,P}\Big[\bH_{n}^{P}(\bx_r)\cdot\bd_r\Big]+W_{m,n}^{S,P}\Big[\bH_{n}^{S}(\bx_r)\cdot \bd_r \Big]\right),\label{MSR-expansionPP}
\\\nm
A^{\parallel,\perp}_{sr}=\sum_{n,m\in\ZZ}
&d_{m}^{P}(s)\left(W_{m,n}^{P,P}\Big[\bH_{n}^{P}(\bx_r)\cdot\bd_r^\perp\Big]+W_{m,n}^{S,P}\Big[\bH_{n}^{S}(\bx_r)\cdot \bd_r^\perp\Big]\right),\label{MSR-expansionPS}
\\\nm
A^{\perp,\parallel}_{sr}=\sum_{n,m\in\ZZ}
&d_{m}^{S}(s)\left(W_{m,n}^{P,S}\Big[\bH_{n}^{P}(\bx_r)\cdot\bd_r\Big]+W_{m,n}^{S,S}\Big[\bH_{n}^{S}(\bx_r)\cdot \bd_r \Big]\right),\label{MSR-expansionSP}
\\\nm
A^{\perp,\perp}_{sr}=\sum_{n,m\in\ZZ}
&d_{m}^{S}(s)\left(W_{m,n}^{P,S}\Big[\bH_{n}^{P}(\bx_r)\cdot\bd_r^\perp\Big]+W_{m,n}^{S,S}\Big[\bH_{n}^{S}(\bx_r)\cdot \bd_r^\perp\Big]\right),\label{MSR-expansionSS}
\end{align}
where $d^P_m(s)= d^P_m(\mathbf{F}_s)$ and $d^S_m(s)=d^S_m(\mathbf{G}_s)$ are the coefficients given by \eqref{d} corresponding to incident fields $\mathbf{F}_s$ and $\mathbf{G}_s$ respectively. Here the parameter $s$ in the argument of $d^\beta_m$ reflects its connection with $s$-th incident field.

Let us now introduce a cut-off parameter $K$ such that the terms with $|n|>K$ or $|m|>K$ are truncated in the expansions \eqref{MSR-expansionPP}-\eqref{MSR-expansionSS} and let $\mathbf{E}^{\ell,\ell'}=(E^{\ell,\ell'}_{sr})\in\CC^{N_s\times N_r}$, for $\ell,\ell'\in\{\parallel,\perp\}$, be the corresponding matrices of truncation errors thus induced. Let us also introduce the matrices $\mathbf{W}^{\alpha,\beta}\in\CC^{(2K+1)\times(2K+1)}$, $\mathbf{X}^{\alpha}\in\CC^{N_s\times (2K+1)}$, $\mathbf{Y}^{\alpha}_\parallel\in\CC^{N_r\times(2K+1)}$ and $\mathbf{Y}^{\alpha}_\perp\in\CC^{N_r\times(2K+1)}$ by
\begin{align*}
& 
\left(\mathbf{W}^{\alpha,\beta}\right)_{mn}:= W^{\alpha,\beta}_{m,n},
\\\nm
&
\left(\mathbf{X}^\beta\right)_{sm}:=d_m^\beta(s),
\\\nm
&
\left(\mathbf{Y}^\alpha_\parallel\right)_{rn}:=\overline{\bH_{n}^{\alpha} (\bx_r)}\cdot \bd_r,
\\\nm
&
\left(\mathbf{Y}^\alpha_\perp\right)_{rn}:=\overline{\bH_{n}^{\alpha}(\bx_r)}\cdot \bd_r^\perp,
\end{align*}
and the block matrices
\begin{align*}
&\mathbf{A}=
\begin{pmatrix}
\mathbf{A}^{\parallel,\parallel} & \mathbf{A}^{\parallel,\perp}
\\\nm
\mathbf{A}^{\perp,\parallel} & \mathbf{A}^{\perp,\perp}
\end{pmatrix}\in\CC^{2N_s\times 2N_r},
\\\nm
&\mathbf{W}=
\begin{pmatrix}
\mathbf{W}^{P,P} & \mathbf{W}^{S,P}
\\\nm
\mathbf{W}^{P,S} & \mathbf{W}^{S,S}
\end{pmatrix}\in\CC^{(4K+2)\times (4K+2)},
\\\nm
&\mathbf{E}=
\begin{pmatrix}
\mathbf{E}^{PP} & \mathbf{E}^{SP}
\\\nm
\mathbf{E}^{PS} & \mathbf{E}^{SS}
\end{pmatrix}\in\CC^{2N_s\times 2N_r},
\\\nm
&\mathbf{X}=
\begin{pmatrix}
\mathbf{X}^{P} & \mathbf{O}_{2K+1}
\\\nm
\mathbf{O}_{2K+1} & \mathbf{X}^{S}
\end{pmatrix}\in \CC^{2N_s\times (4K+2)},
\\\nm
&\mathbf{Y}=
\begin{pmatrix}
\mathbf{Y}^{P}_\parallel & \mathbf{Y}^{S}_\parallel
\\\nm
\mathbf{Y}^{P}_\perp & \mathbf{Y}^{S}_\perp
\end{pmatrix}\in\CC^{2N_r\times (4K+2)},
\end{align*}
where $\mathbf{O}_{2K+1}\in\RR^{N_s\times(2K+1)}$ is the zero matrix. It can be seen after fairly easy manipulations that the global MSR matrix can be expressed as
$$
\mathbf{A}=\mathbf{X}\mathbf{W}\mathbf{Y}^*+\mathbf{E},
$$
where $*$ reflects the Hermitian transpose of a matrix, i.e., $\mathbf{A}^*=\overline{\mathbf{A}}^\top$.

The following result is readily proved thanks to Lemma \ref{Lem:Wsym2}.
\begin{lem}\label{Lem:Wsym3}
The global block matrix $\mathbf{W}$ is Hermitian, i.e, $\mathbf{W}= \mathbf{W}^*$.
\end{lem}

\subsection{Least-squares Minimization Algorithm}\label{ss:least}

Let us define the linear transformation $\mathbf{L}:\CC^{(4K+2)\times(4K+2)}\to \CC^{2N_s\times 2N_r}$ by
$$
\mathbf{L}(\mathbf{M}):=\mathbf{X}\mathbf{M}\mathbf{Y}^*,
$$
and let $\mathbf{N}_{\noise}\in\CC^{2N_s\times 2N_r}$ denote the measurement noise. For simplicity, we assume that each entry $(\mathbf{N}_{\noise})_{sr}$ is an independent and identically distributed complex random noise with mean zero and variance $\sigma^2_{\noise}$ such that
$$
\mathbf{N}_{\noise}=\sigma_{\noise}\mathbf{N}_0\quad\text{with}\quad (\mathbf{N}_0)_{sr}\sim\mathcal{N}(0,1).
$$

In this subsection, we consider the noisy measurements
\begin{align}
\mathbf{A}=\mathbf{X}\mathbf{W}\mathbf{Y}^*+\mathbf{E}+\mathbf{N}_{\noise}=\mathbf{L}(\mathbf{W})+\mathbf{E}+\mathbf{N}_{\noise},
\label{Lsystem}
\end{align}
and design a procedure to retrieve the solution $\mathbf{W}$. Let us reconstruct a least-squares minimization solution for the linear system \eqref{Lsystem} in $\ker\mathbf{L}^\perp$ by
\begin{align}
\widehat{\mathbf{W}}:=\argmin_{\mathbf{M}\in\ker\mathbf{L}^\perp}\|\mathbf{L}(\mathbf{M})-\mathbf{A}\|_F,
\label{eq:LeastSq}
\end{align}
where $\|\cdot\|_F$ denotes the Frobenius norm of matrices and $\ker\mathbf{L}$ denotes the kernel of the linear operator $\mathbf{L}$. Note that if the cut-off parameter $K$ is such that $(2K+1)<N_r,N_s$ and both matrices $\mathbf{X}$ and $\mathbf{Y}$ are full rank  then $\mathbf{L}$ is rank preserving and $\ker\mathbf{L}$ is trivial. Consequently, the admissible set for the least-squares minimization turns out to be $\RR^{(4K+2)\times(4K+2)}$ and $\widehat{\mathbf{W}}$ can be explicitly calculated in the absence of measurement noise. In this case, $\mathbf{X}^\alpha$ is a Fourier matrix by virtue of \eqref{UP}, \eqref{US} and \eqref{U-multipole-Co} and
\begin{eqnarray}
\left(\mathbf{X}^\alpha\right)^*\mathbf{X}^\alpha= \frac{N_s}{|b_\alpha|^2}\mathbf{I}_{2K+1}
\quad\text{with}\quad
b_\alpha= {(2\rho_0\omega  c_\beta)^2\K_\beta}, \label{b_alpha}
\end{eqnarray}
where $\mathbf{I}_{2K+1}\in\RR^{(2K+1)\times(2K+1)}$ is the identity matrix. Consequently,
\begin{eqnarray}
\mathbf{X}^*\mathbf{X}=
\begin{pmatrix}
\left(\mathbf{X}^P\right)^*\mathbf{X}^P & \mathbf{O}_{2K+1}
\\\nm
\mathbf{O}_{2K+1} & \left(\mathbf{X}^S\right)^*\mathbf{X}^S
\end{pmatrix}
=N_s\mathbf{Z}_{\mathbf{X}},
\label{OrthogonalityX}
\end{eqnarray}
with
\begin{align*}
\mathbf{Z}_{\mathbf{X}}:=
\begin{pmatrix}
|b_P|^{-2}\mathbf{I}_{2K+1} & \mathbf{O}_{2K+1}
\\\nm
\mathbf{O}_{2K+1} & |b_S|^{-2}\mathbf{I}_{2K+1}
\end{pmatrix}.
\end{align*}
Note also that
\begin{align*}
\left(\mathbf{Y}^\alpha_\parallel\right)^*\mathbf{Y}^\beta_\parallel= N_r\mathbf{C}^{\alpha,\beta}
\quad\text{and}\quad
\left(\mathbf{Y}^\alpha_\perp\right)^*\mathbf{Y}^\beta_\perp= N_r\mathbf{D}^{\alpha,\beta},
\end{align*}
where $\mathbf{C}^{\alpha,\beta},\mathbf{D}^{\alpha,\beta}\in\RR^{(2K+1)\times(2K+1)}$ are diagonal matrices
\begin{align*}
\mathbf{C}^{\alpha,\beta}:=&
{\rm diag}\begin{pmatrix}
g^\alpha_{-K} \overline{g^\beta_{-K}},
& \cdots,
&  g^\alpha_{K}\overline{g^\beta_{K}},
\end{pmatrix},
\\\nm
\mathbf{D}^{\alpha,\beta}:=&
{\rm diag}
\begin{pmatrix}
h^\alpha_{-K} \overline{h^\beta_{-K}},
&  \cdots,   & h^\alpha_{K}\overline{h^\beta_{K}}
\end{pmatrix},
\end{align*}
with
\begin{align}
g_m^P:=&\K_P\left(H^{(1)}_m(\K_P R)\right)'
\quad\text{and}\quad
g_m^S:=\frac{im}{R}H^{(1)}_m(\K_S R),\label{g}
\\\nm
h_m^P:=&\frac{im}{R}H^{(1)}_m(\K_P R)
\quad\text{and}\quad
h_m^S:=-\K_S\left(H^{(1)}_m(\K_S R)\right)'.\label{h}
\end{align}
Therefore,
\begin{align*}
\left(\mathbf{Y}\right)^*\mathbf{Y}&=
\begin{pmatrix}
\ds\left(\mathbf{Y}^{P}_\parallel\right)^*\mathbf{Y}^{P}_\parallel+\left(\mathbf{Y}^{P}_\perp\right)^*\mathbf{Y}^{P}_\perp
&&
\ds\left(\mathbf{Y}^{P}_\parallel\right)^*\mathbf{Y}^{S}_\parallel+\left(\mathbf{Y}^{P}_\perp\right)^*\mathbf{Y}^{S}_\perp
\\\nm
\ds\left(\mathbf{Y}^{S}_\parallel\right)^*\mathbf{Y}^{P}_\parallel+\left(\mathbf{Y}^{S}_\perp\right)^*\mathbf{Y}^{P}_\perp
&&
\ds\left(\mathbf{Y}^{S}_\parallel\right)^*\mathbf{Y}^{S}_\parallel+\left(\mathbf{Y}^{S}_\perp\right)^*\mathbf{Y}^{S}_\perp
\end{pmatrix}
\nonumber
\\\nm
&=
N_r
\begin{pmatrix}
\mathbf{C}^{P,P}+\mathbf{D}^{P,P}
& \mathbf{C}^{P,S}+\mathbf{D}^{P,S}
\\\nm
\mathbf{C}^{S,P}+\mathbf{D}^{S,P}
&
\mathbf{C}^{S,S}+\mathbf{D}^{S,S}
\end{pmatrix}.
\end{align*}

It can be easily proved that $\mathbf{Y}^*\mathbf{Y}$ becomes diagonal when the radius $R$ of the imaging domain $\partial B_R(\textbf{0})$ is sufficiently large. Precisely, the following result holds.
\begin{lem}\label{LemY}
For the radius $R$ of the ball $B_R(\textbf{0})$ approaching to infinity  the matrix $\mathbf{Y}^*\mathbf{Y}$ admits a decomposition
\begin{align*}
\mathbf{Y}^*\mathbf{Y}=N_r\mathbf{Z}_{\mathbf{Y}}+\mathbf{Q}
\quad\text{with}\quad
\mathbf{Z}_{\mathbf{Y}}:=
\begin{pmatrix}
\mathbf{C}^{P,P} & \mathbf{O}_{2K+1}
\\\nm
\mathbf{O}_{2K+1} & \mathbf{D}^{S,S}
\end{pmatrix},
\end{align*}
where $\mathbf{Q}=\left(q_{\ell\ell'}\right)_{\ell,\ell'=1\cdots, 4K+2}$ is such that $|q_{\ell\ell'}|\leq CR^{-2}$ for some constant $C\in\RR_+$ independent of $R$.
\end{lem}

\begin{proof}
Let $C$ denote a generic constant that may vary at each step. Note that the matrix $\mathbf{Y}^*\mathbf{Y}$ can be decomposed as
\begin{eqnarray}
\mathbf{Y}^*\mathbf{Y}:=
N_r
\begin{pmatrix}
\mathbf{C}^{P,P} & \mathbf{O}_{2K+1}
\\\nm
\mathbf{O}_{2K+1} & \mathbf{D}^{S,S}
\end{pmatrix}
+
N_r\begin{pmatrix}
\mathbf{D}^{P,P} & \mathbf{C}^{P,S}+\mathbf{D}^{P,S}
\\\nm
\mathbf{C}^{S,P}+\mathbf{D}^{S,P} & \mathbf{C}^{S,S}
\end{pmatrix}.\label{Y*Y-diagonal}
\end{eqnarray}
Recall that $\mathbf{C}^{\alpha,\beta}$ and $\mathbf{D}^{\alpha,\beta}$ are diagonal matrices and in particular
$$
\left(\mathbf{C}^{S,S}\right)_{mm}=\frac{m^2}{R^2}\left| H^{(1)}_m(\K_SR)\right|^{2}
\quad\text{and}\quad
\left(\mathbf{D}^{P,P}\right)_{mm}=\frac{m^2}{R^2}\left|H^{(1)}_m(\K_PR)\right|^{2}.
$$
By virtue of the decay property \eqref{decayH} of $H^{(1)}_m$, as $R\to+\infty$,  we have
$$
\left|\left(\mathbf{C}^{S,S}\right)_{mn}\right|\leq\frac{C}{R^{3}}
\quad\text{and}\quad
\left|\left(\mathbf{D}^{P,P}\right)_{mn}\right|\leq \frac{C}{R^{3}}.
$$
Similarly, the decay properties \eqref{decayH}-\eqref{decayHPrime} furnish
\begin{align*}
\left|\left(\mathbf{C}^{P,S}\right)_{mn}\right|\leq & \frac{C}{R^{2}}
\quad\text{and}\quad
\left|\left(\mathbf{C}^{S,P}\right)_{mn}\right|\leq   \frac{C}{R^{2}},
\\\nm
\left|\left(\mathbf{D}^{P,S}\right)_{mn}\right|\leq & \frac{C}{R^{2}}
\quad\text{and}\quad
\left|\left(\mathbf{D}^{S,P}\right)_{mn}\right|\leq   \frac{C}{R^{2}},
\end{align*}
as $R\to+\infty$. This shows the decay of the elements of second matrix on  right hand side (RHS) of \eqref{Y*Y-diagonal}, which leads to the required form of $\mathbf{Y}^*\mathbf{Y}$ for $R\to+\infty$.
\end{proof}

An important consequence of Lamma \ref{LemY} and the orthogonality relation \eqref{OrthogonalityX} is the following result  substantiating that the linear operator $\mathbf{L}$ possesses a left pseudo-inverse when $R\to+\infty$.
\begin{thm}\label{thmLInv}
If $ 2K+1 \leq N_r,N_s$ and matrices $\mathbf{X}$ and $\mathbf{Y}$ are  full-rank then the linear operator $\mathbf{L}:\CC^{(4K+2)\times(4K+2)}\to\CC^{2N_s\times2N_r}$ possesses a left pseudo-inverse
$$
\mathbf{L}^\dagger\left(\mathbf{A}\right):=\frac{1}{N_sN_r}\mathbf{Z}^{-1}_{\mathbf{X}}\mathbf{X}^*\mathbf{A}\mathbf{Y}\mathbf{Z}^{-1}_{\mathbf{Y}},
$$
when $R\to+\infty$.
\end{thm}
\begin{proof}
Since $(2K+1)\leq N_s,N_r$, $\mathbf{X}$ and $\mathbf{Y}$ are full-rank, and $R\to+\infty$, it is easy to see that both $\mathbf{X}$ and $\mathbf{Y}$ possess left pseudo-inverses, denoted by $\mathbf{X}^\dagger$ and $\mathbf{Y}^\dagger$ respectively, thanks to the orthogonality property \eqref{OrthogonalityX} and Lemma \ref{LemY}. Precisely,
$$
\mathbf{X}^\dagger:=\left(\mathbf{X}^*\mathbf{X}\right)^{-1}\mathbf{X}^*=\frac{1}{N_s}\mathbf{Z}^{-1}_{\mathbf{X}}\mathbf{X}^*
\quad\text{and}\quad
\mathbf{Y}^\dagger:=\left(\mathbf{Y}^*\mathbf{Y}\right)^{-1}\mathbf{Y}^*=\frac{1}{N_r}\mathbf{Z}^{-1}_{\mathbf{Y}}\mathbf{Y}^*,
$$
as $R\to+\infty$. Consequently, we have
\begin{align*}
\frac{1}{N_sN_r}\mathbf{Z}^{-1}_{\mathbf{X}}\mathbf{X}^*\mathbf{A}\mathbf{Y}\mathbf{Z}^{-1}_{\mathbf{Y}}
=&
\frac{1}{N_r}\left(\mathbf{X}^*\mathbf{X}\right)^{-1}\mathbf{X}^*\left(\mathbf{X} \widehat{\mathbf{W}}\mathbf{Y}^*\right)\mathbf{Y}\mathbf{Z}^{-1}_{\mathbf{Y}}
\\
=&
\frac{1}{N_r} \widehat{\mathbf{W}}\left(\mathbf{Y}^*\mathbf{Y}\right)\mathbf{Z}^{-1}_{\mathbf{Y}}
= \widehat{\mathbf{W}}.
\end{align*}
This completes the proof.
\end{proof}

\subsection{Stability Analysis}\label{ss:stability}

In this section,  we perform a stability analysis for the linear operator $\mathbf{L}$.  We substantiate that the operator $\mathbf{L}$ is ill-conditioned for $K\to +\infty$. It simply means that only a certain number of lower order scattering coefficients can be recovered stably  which in turn contain only lower order information of the shape oscillations of boundary $\partial D$. The limit on the information about the shape and morphology of the inclusion $D$ that can be obtained stably is determined by the maximal resolving order and the stability estimate for the operator $\mathbf{L}$  thereby defining the resolution limit of the imaging paradigm.
Towards this end, the following result characterizes the singular values and the singular vectors of the operator $\mathbf{L}$.
\begin{thm}\label{thmSV} If $N_s,N_r\geq 2K+1$ and $R\to+\infty$ then the right singular vectors of $\mathbf{L}$ are coincident with the canonical basis of $\RR^{(4K+2)\times(4K+2)}$ and the $(p,q)-$th singular value  of the operator $\mathbf{L}$ is given by
\begin{align}
\sigma_{pq}:=
\begin{cases}
\ds\frac{\sqrt{N_sN_r}}{4\rho_0^2\omega^2c_P^2\K_P}\left|\left(H^{(1)}_{q-1-K}(\K_P R)\right)'\right|, & 1\leq p,q\leq 2K+1,
\\\nm
\ds\frac{\sqrt{N_sN_r}}{4\rho_0^2\omega^2c_S^2\K_S}\left|\left(H^{(1)}_{q-2-3K}(\K_S R)\right)'\right|, & 2K+2\leq p,q\leq 4K+2.
\end{cases}
\label{sigma_mn}
\end{align}
\end{thm}
\begin{proof}
Let us define the inner product of two complex matrices $\mathbf{N}$ and $\mathbf{M}$ by
$$
\langle \mathbf{N},\mathbf{M}\rangle :=\sum_{\ell,\ell'}\left(\mathbf{N}^*\right)_{\ell \ell'}\left(\mathbf{M}\right)_{\ell \ell'}.
$$
Let $\mathbf{V}_{pq}\in\RR^{(4K+2)\times(4K+2)}$, for each $p,q=1,2,\cdots,4K+2$, be such that
$$
\left(\mathbf{V}_{pq}\right)_{\ell\ell'}=\delta_{p\ell}\delta_{q\ell'},\quad\forall\, \ell,\ell'=1,2,\cdots,4K+2.
$$
It is easy to verify that for $R\to+\infty$, thanks to diagonality result \eqref{OrthogonalityX} and Lemma \ref{LemY},
\begin{align*}
\langle \mathbf{L}\left(\mathbf{V}_{pq}\right),\mathbf{L}\left(\mathbf{V}_{p'q'}\right)\rangle
=&
\langle \mathbf{X}\mathbf{V}_{pq}\mathbf{Y}^*,\mathbf{X}\mathbf{V}_{p'q'}\mathbf{Y}^*\rangle
\\
=&
N_sN_r\langle \mathbf{V}_{pq},\mathbf{Z}_\mathbf{X}\mathbf{V}_{p'q'}\mathbf{Z}_\mathbf{Y}\rangle
\\
=&\delta_{pp'}\delta_{qq'}N_sN_r|f_q|^2,
\end{align*}
where
$$
|f_q|:=\begin{cases}
\ds\frac{|g_{q-1-K}^{P}|}{|b_P|}, & 1\leq p,q\leq 2K+1,
\\\nm
\ds\frac{|h_{q-2-3K}^{S}|}{|b_S|}, & 2K+2\leq p,q\leq 4K+2.
\end{cases}
$$
On substituting the values of $g^P_{q-1-K}$, $h^S_{q-2-3K}$  and $b_\alpha$ from \eqref{g}, \eqref{h} and \eqref{b_alpha}, we arrive  at
$$
|f_q|:=\begin{cases}
\ds\frac{1}{4\rho_0^2\omega^2c_P^2\K_P}\left|\left(H^{(1)}_{q-1-K}(\K_P R)\right)'\right|, & 1\leq p,q\leq 2K+1,
\\\nm
\ds\frac{1}{4\rho_0^2\omega^2c_S^2\K_S}\left|\left(H^{(1)}_{q-2-3K}(\K_S R)\right)'\right|, & 2K+2\leq p,q\leq 4K+2.
\end{cases}
$$
This shows that the canonical basis $\big\{\mathbf{V}_{pq}\big\}_{p,q=1,\cdots,4K+2}$ forms the set of right singular vectors of $\mathbf{L}$ and the $(p,q)$-th  singular value of the operator $\mathbf{L}$ is thus rendered by $\left\|\mathbf{L}\left(\mathbf{V}_{pq}\right)\right\|_{F}$ and is given by \eqref{sigma_mn}. Moreover, the left singular vectors of $\mathbf{L}$ are furnished by the relation $\widetilde{\mathbf{V}}_{pq}:=\mathbf{L}\left(\mathbf{V}_{pq}\right)/\sigma_{pq}$.
\end{proof}

It should be observed that the quantities $|g^P_{2K+1}|$ and $|h^S_{2K+1}|$ diverge when $K\to+\infty$. Consequently, the operator $\mathbf{L}$ is unbounded. Indeed, we have the following estimate for the condition number of $\mathbf{L}$ thanks to Theorem \ref{thmSV}.

\begin{cor}\label{CorCond}
Under the assumptions of Theorem \ref{thmSV},
$$
{\rm cond}\left(\mathbf{L}\right)\lesssim \left(C_R^P K\right)^{(K+1)},
\quad\text{as}\quad K\to+\infty,
$$
where $C_R^\alpha:=\ds{2}/{e\K_\alpha R}.$
\end{cor}

\begin{proof}
Let  $\sigma_{\rm max}$ and  $\sigma_{\rm min}$ be the largest and the smallest singular values of the operator $\mathbf{L}$. Recall the asymptotic behavior of the Bessel functions of first and second kind
\begin{equation*}
J_m(t)\simeq \sqrt{\frac{1}{2\pi |m|}}\left(\frac{et}{2|m|}\right)^{|m|}
\quad\text{and}\quad
Y_m(t)\simeq -\sqrt{\frac{2}{\pi |m|}}\left(\frac{et}{2|m|}\right)^{-|m|},
\end{equation*}
with respect to the order $|m|\to+\infty$ at a fixed argument $t$ (see, for instance, \cite[Formulae 10.19.1 and 10.19.2]{NIST}). Consequently,  an easy commutation shows that
$$
\left|H^{(1)}_m(\K_\alpha R)\right|\lesssim \left(C_R^\alpha |m|\right)^{|m|}+\left(C_R^\alpha|m|\right)^{-|m|}, \quad\text{as }\,|m|\to+\infty.
$$

Moreover, invoking the recurrence relation (see, for instance, \cite[Formula 10.6.2]{NIST})
$$
\left(H^{(1)}_m(t)\right)'=H^{(1)}_{m-1}(t)-\frac{m}{t}H^{(1)}_m(t),
$$
it is easy to get that
\begin{align*}
\left|\left(H^{(1)}_m(\K_\alpha R)\right)'\right|
\lesssim &
(C_R^\alpha(m-1))^{(m-1)}+(C_R^\alpha(m-1))^{-(m-1)}
+\frac{e}{2}C_R^\alpha{m}\left((C_R^\alpha m)^{m}+(C_R^\alpha m)^{-m}\right)
\nonumber
\\
\lesssim &
\left(C_R^\alpha m\right)^{m+1},
\end{align*}
when $|m|\to+\infty$. Consequently,
\begin{align*}
\sigma_{(2K+1)(2K+1)}\lesssim \left(C_R^PK\right)^{K+1}
\quad\text{and}\quad
\sigma_{(4K+2)(4K+2)}\lesssim \left(C_R^S K\right)^{K+1}.
\end{align*}
Finally,   note that the relation $\sigma_{\rm max}\simeq \sigma_{(2K+1)(2K+1)}$ holds when $K$ is large enough,  which follows from the fact that $C_R^P> C_R^S$ (this is due to the inequality $c_P>c_S$, since $\mu_0, \lambda_0>0$). Moreover, the smallest singular value $\sigma_{\rm min}$ is bounded. Therefore,
\begin{eqnarray*}
{\rm cond}(\mathbf{L})=\frac{\sigma_{\rm max}}{\sigma_{\rm min}}\lesssim \left(C_R^PK\right)^{ K+1 }.
\end{eqnarray*}
\end{proof}

It view of the aforementioned result, the least squares minimization problem \eqref{eq:LeastSq} turns out to be highly ill-conditioned. However, this ill-posedness can be reduced by considering the constrained optimization problem subject to the energy conservation constraint \eqref{W-constraint}.

\subsection{Error Analysis}\label{ss:error}

Let us now analyze the error committed by truncating the infinite series in the MSR data. But  before further discussion, we recall the following result from \cite[Appendix A]{Shape}.
\begin{lem}\label{LemErr}
For $c>0$ and $N\in\mathbb{N}$ such that $N>c/e$
$$
\sum_{n>N}\left(\frac{c}{n}\right)^n\leq \left(\frac{c}{N}\right)^N\left(\frac{1}{1+\ln(N/c)}\right).
$$
\end{lem}

The following is the main result of this section.
\begin{thm}\label{thmErr}
Let $C^\alpha_R$ and $C_{\alpha,\beta}>1$ be the constants defined in Corollary \ref{CorCond} and Lemma \ref{LemWbound} respectively. Let the radius $R$ of the measurement domain $B_R(\textbf{0})$ be such that $h:=\max_{\alpha,\beta}\{2C_{\alpha,\beta}^2C^\alpha_R\}<1$. Then  there exists a sufficiently large truncation order $K$ satisfying $K>\ds\max_{\alpha,\beta}\{C_{\alpha,\beta}/(C_R^\alpha e) \}$  such that
$$
|E^{\alpha,\beta}_{sr}|=O(h^{K-1}).
$$
\end{thm}
\begin{proof}
The result for the truncation error $E^{P,P}_{sr}$ is proved only. The rest of the estimates can be obtained following the same procedure. First, split the summations into three different contributions as
\begin{align*}
E^{P,P}_{sr}=&\left(\sum_{\substack{|m|\leq K\\|n|>K}}+\sum_{\substack{|m|> K\\|n|\leq K}}+\sum_{\substack{|m|> K\\|n|>K}}\right)d_{m}^{P}(s)\left(W_{m,n}^{P,P}\Big[\bH_{n}^{P}(\bx_r)\cdot\bd_r\Big]+W_{m,n}^{S,P}\Big[\bH_{n}^{S}(\bx_r)\cdot \bd_r \Big]\right)
\nonumber
\\
=&I_1+I_2+I_3.
\end{align*}
Then, thanks to Lemma \ref{LemWbound} and invoking the  definitions \eqref{d} and \eqref{UP}-\eqref{US} of $d_m^P$ and $\bH^\alpha_n$ respectively, we have
\begin{align*}
|I_1|\leq & \frac{1}{4\rho_0^2\omega^2c^2_P\K_P}\Bigg(\sum_{|m|\leq K}\frac{C_{P,P}^{|m|-1}}{|m|^{|m|-1}}\sum_{|n|>K}\frac{C_{P,P}^{|n|-1}}{|n|^{|n|-1}}\left|\K_P \left(H^{(1)}_n(\K_P R)\right)'\right|
\nonumber
\\
&
\qquad\qquad\qquad\qquad
+\sum_{|m|\leq K}\frac{C_{S,P}^{|m|-1}}{|m|^{|m|-1}}\sum_{|n|>K}\frac{C_{S,P}^{|n|-1}}{|n|^{|n|-1}}\frac{|n|}{R}\left|H^{(1)}_n(\K_S R)\right|\Bigg).
\end{align*}
Recall the estimates
\begin{align*}
&\left|H^{(1)}_n(\K_\alpha R)\right|
\lesssim 
\left(C_R^\alpha |n|\right)^{|n|}+\left(C_R^\alpha|n|\right)^{-|n|},
\\
&\left|\left(H^{(1)}_{n}(\K_\alpha R)\right)'\right|
\lesssim 
\frac{e}{2}(C_R^\alpha |n|)^{|n|+1}+ (C_R^\alpha (|n|-1))^{|n|-1}
\\
&\qquad\qquad\qquad\qquad
+\frac{e}{2}(C_R^\alpha |n|)^{1-|n|}+(C_R^\alpha (|n|-1))^{1-|n|},
\end{align*}
as $|n|\to+\infty$ and note that, up to some factors independent of $K$,
\begin{align*}
\sum_{|m|\leq K}\frac{C_{P,P}^{|m|-1}}{|m|^{|m|-1}}\lesssim C^{K-1}_{P,P}
\quad\text{and}\quad
\sum_{|m|\leq K}\frac{C_{S,P}^{|m|-1}}{|m|^{|m|-1}}\lesssim C^{K-1}_{S,P}.
\end{align*}
Therefore,
\begin{align}
|I_1|\lesssim &
\frac{C_{P,P}^{K-1}}{4\rho_0^2\omega^2c^2_P}
\sum_{|n|>K}\Bigg[
\frac{e}{2}\left(C_R^P|n|\right)^2\left(C_{P,P}C^P_R\right)^{|n|-1}
+
\left(1-\frac{1}{|n|}\right)^{|n|-1}\left(C_{P,P}C_R^P\right)^{|n|-1}
\nonumber
\\
&\qquad\qquad\qquad\qquad\qquad\qquad
+\frac{e}{2}\left(\frac{C_{P,P}/C_R^P}{|n|^2}\right)^{|n|-1}
+
\left(\frac{C_{P,P}/C^P_R}{|n|(|n|-1)}\right)^{|n|-1}\Bigg]
\nonumber
\\
&+
\frac{e \K_{S} C_{S,P}^{K-1}}{8\rho_0^2\omega^2c^2_P\K_P}
\sum_{|n|>K}\Bigg[
\left(C_R^S|n|\right)^2\left(C_{S,P}C^S_R\right)^{|n|-1}
+
\left(\frac{C_{S,P}/C^S_R}{|n|^2}\right)^{|n|-1}\Bigg].\label{I1_mid}
\end{align}
Thanks to Lemma \ref{LemErr}, the third, fourth and sixth terms on RHS of \eqref{I1_mid} are negligible for all $K> \max_{\alpha,\beta}\{C_{\alpha,\beta}/C_R^\alpha e\}$. Moreover, it can be easily verified that
\begin{align}
\frac{n^2}{2^{n+2}}\leq 1
\quad\text{and}\quad
\frac{1}{2^{n-1}}\left(1-\frac{1}{n}\right)^{n-1}\leq 1,\quad \forall\;  n\in\NN, \,n>1,
\end{align}
and
\begin{align}
\left(C^\alpha_R\right)^2<C^\alpha_R<{C_{\alpha,\beta} C^\alpha_R}<{C^2_{\alpha,\beta} C^\alpha_R} <\max_{\alpha,\beta}\{C_{\alpha,\beta}^2C^\alpha_R\}<\frac{1}{2}.
\end{align}
Therefore, we have
\begin{align}
&\frac{e C^{K-1}_{P,P}}{8}\sum_{|n|>K} {\left(C^P_R|n|\right)^2}\left(C_{P,P} C^P_R\right)^{|n|-1}
\nonumber
\\
&\qquad 
\lesssim 
e C^{K-1}_{P,P} \sum_{|n|>K}\frac{|n|^2}{2^{|n|+2}}\left(2C_{P,P} C^P_R\right)^{|n|-1}
\lesssim   
e C^{K-1}_{P,P}\left(2C_{P,P} C^P_R\right)^{K-1}
\leq    \,\,e h^{K-1},
\label{est:1}
\\\nm
&\frac{C^{K-1}_{P,P}}{4}\sum_{|n|>K}
 \left(1-\frac{1}{|n|}\right)^{|n|-1}\left(C_{P,P} C^P_R\right)^{|n|-1}
\nonumber
\\
&\qquad 
\lesssim  
{C^{K-1}_{P,P}}\sum_{|n|>K} \frac{1}{2^{|n|-1}}\left(1-\frac{1}{|n|}\right)^{|n|-1}\left(2C_{P,P} C^P_R\right)^{|n|-1}
\lesssim  
C^{K-1}_{P,P}\left(2C_{P,P} C^P_R\right)^{K-1}
\leq   \,\,
h^{K-1},
\label{est:2}
\\\nm
&\frac{e\K_S C^{K-1}_{S,P}}{8\K_P}\sum_{|n|>K}\left(C^S_R|n|\right)^2\left(C_{S,P} C^S_R\right)^{|n|-1}
\nonumber
\\
&\qquad 
\lesssim 
\frac{c_P C^{K-1}_{S,P}}{c_S}\sum_{|n|>K}\frac{|n|^2}{2^{|n|+2}}\left(2C_{S,P}C^S_R\right)^{|n|-1}
\lesssim 
\frac{ c_P C^{K-1}_{S,P}}{c_S}\left(2C_{S,P}C^S_R\right)^{K-1}
\leq  \,\,
\frac{c_P}{c_S} h^{K-1}.
\label{est:3}
\end{align}
Substituting the estimates \eqref{est:1}-\eqref{est:3} in \eqref{I1_mid} one arrives at
\begin{align*}
|I_1|\lesssim \frac{1}{\rho_0\omega^2c^2_P}\left(e+ 1+\frac{c_P}{c_S} \right)h^{K-1}.
\end{align*}
The estimate for $|I_2|$ follows by changing the role of $m$ and $n$. Moreover, by proceeding in a similar fashion, it can be easily established that
$$
|I_3|\lesssim \left(\frac{h}{K}\right)^{K-1}.
$$
Combining the estimates for $|I_1|$, $|I_2|$ and $|I_3|$, one obtains $|E^{P,P}_{sr}|\lesssim h^{K-1}$. This completes the proof.

\end{proof}

\subsection{Maximal Resolving Order}\label{ss:snr}

In this subsection, we determine the maximal resolving order $K$ for the reconstruction framework. In order to do so, we first estimate the \emph{strength} of the recorded signals in terms of the geometry of inclusion $D$ and the radius of the recording circle. Then we define the signal-to-noise ratio (SNR) in terms of signal strength and noise standard deviation $\sigma_{\rm noise}$. Towards this end, it is easy to see from the integral representation \eqref{u-int-rep} that
$$
{A}^{\parallel,\parallel}_{sr}=
\mathcal{S}^\omega_D[\bpsi_{\bF_s}](\bx_r)\cdot\bd_r=
\mathcal{S}^\omega_D[\bpsi_{\bF_s}](R\,\bd_r)\cdot\bd_r,
$$
where
$\bpsi_{\bF_s}$ is the solution of \eqref{integral-system} corresponding to $\bu^{\rm inc}=\bF_s$. By virtue of the far field behavior,
\begin{align*}
\bGamma^\omega(\bx,\by)
\simeq&
\frac{e^{i\K_P|\bx|}}{\sqrt{|x|}}
\left(\frac{i+1}{4\rho_0 c_P^2\sqrt{\pi\K_P}} \hbx\otimes\hbx e^{-i\K_P\hbx\cdot\by}\right)
\nonumber
\\
&+
\frac{e^{i\K_S|\bx|}}{\sqrt{|x|}}\left(\frac{i+1}{4\rho_0 c_S^2\sqrt{\pi\K_S}} \left(\mathbf{I}_2-\hbx\otimes\hbx\right)e^{-i\K_S\hbx\cdot\by}\right),
\end{align*}
of the fundamental solution for a bounded $\by\in\RR^2$ and  $\bx\in\RR^2$ such that $|\bx|\to+\infty$,  one has
\begin{align}
A^{\parallel,\parallel}\simeq &
\frac{1}{\sqrt{R}}\frac{(i+1)e^{i\K_P R}}{4\rho_0 c_P^2\sqrt{\pi\K_P}}\int_{\partial D}[(\bd_r\otimes\bd_r)\bpsi_{\bF_s}(\by)]\cdot\bd_r\, e^{-i\K_P|\by|\cos(\theta_r-\theta_\by)}\,d \sigma(\by)&
\nonumber
\\
\qquad
& +\frac{1}{\sqrt{R}}\frac{(i+1)e^{i\K_S R}}{4\rho_0c_S^2\sqrt{\pi\K_S}}\int_{\partial D}[\bd^\perp_r\otimes\bd^\perp_r)\bpsi_{\bF_s}(\by)]\cdot \bd_r\,e^{-i\K_S|\by|\cos(\theta_r-\theta_\by)}\, d\sigma(\by) 
\nonumber 
\\\nm
\simeq &
\frac{1}{\sqrt{R}}\frac{(i+1)e^{i\K_P R}}{4\rho_0 c_P^2\sqrt{\pi\K_P}}\int_{\partial D}[\bpsi_{\bF_s}(\by)\cdot\bd_r]\, e^{-i\K_P|\by|\cos(\theta_r-\theta_\by)}\,d \sigma(\by)&
\nonumber
\\
\qquad
& +\frac{1}{\sqrt{R}}\frac{(i+1)e^{i\K_S R}}{4\rho_0c_S^2\sqrt{\pi\K_S}}\int_{\partial D}[\bpsi_{\bF_s}(\by)\cdot \bd_r^\perp]\,e^{-i\K_S|\by|\cos(\theta_r-\theta_\by)}\, d\sigma(\by). \label{AFarField}
\end{align}
On the other hand, by \eqref{stability}
$$
\left\|\bpsi_{\bF_s}\right\|_{L^2(\partial D)^2}\leq \left\|\bF_s\right\|_{H^1(\partial D)^2}+\left\|\frac{\partial\bF_s}{\partial \nu}\right\|_{L^2(\partial D)^2}\leq C \sqrt{|\partial D|},
$$
for some constant independent of $R$ and $|\partial D|$.
Thus, by taking the modulus on both sides of \eqref{AFarField}, substituting the above estimate for $\|\bpsi_{\bF_s}\|$, and using the Cauchy-Schwartz inequality, one obtains the estimate
$$
|A^{\parallel,\parallel}|\leq C \frac{|\partial D|}{\sqrt{R}}.
$$
The constant $C$ above depends only on the material parameters of the background domain, inclusion $D$, and the frequency $\omega$ of the incident field but is  independent of $R$ and $\partial D$. Similarly, the terms of other MSR matrices can be also bounded by $|\partial D|/\sqrt{R}$. This endorses that the measured signals are of order $|\partial D|/\sqrt{R}$. Therefore, we define SNR by
$$
\snr:= \frac{|\partial D|/\sqrt{R}}{\sigma_{\noise}}.
$$

Now we are ready to estimate the maximal resolving order $K$. In the sequel, $\mathbb{E}$ denotes the expectation with respect to the statistics of the
noise $\mathbf{N}_{\noise}$. Moreover, we work in the regime when $R\to+\infty$ (or $O(R^{-3/2})$ terms are negligible) and the truncation error is much smaller than the noise standard deviation, which in turn is much smaller than the order of the signal (or simply SNR is much larger than $1$), that is,
\begin{align}
h^{K-1}\ll \sigma_{\noise}\ll |\partial D|/\sqrt{R}.\label{NoiseRegime}
\end{align}
From the injectivity of operator $\mathbf{L}$ for $R\to+\infty$ and the relation \eqref{Lsystem},  we have
\begin{align}
\mathbb{E}\left(\left|\left(\widehat{\mathbf{W}}-\mathbf{W}\right)_{mn}\right|^2\right)^{1/2}
= & \,\,
\mathbb{E}\left(\Big|(\mathbf{L}^\dagger(\mathbf{E}+\mathbf{N}_\noise)_{mn}\Big|^2\right)^{1/2}
\nonumber
\\
\leq & \,\,
\Big|\mathbf{L}^\dagger(\mathbf{E})_{mn}\Big|+\Big|\mathbf{L}^\dagger(\mathbf{N}_\noise)_{mn}\Big|
\nonumber
\\\nm
\lesssim & \,\,
\sigma^{-1}_{mn}\left(\left\|\mathbf{E}\right\|_{F}+\left\|\mathbf{N}_\noise\right\|_{F}\right)
\nonumber
\\\nm
\lesssim & \,\,
\sigma^{-1}_{mn}\left(h^{K-1}+ \sigma_\noise\sqrt{N_sN_r}\right),\label{E_mid}
\end{align}
where Cauchy-Schwartz inequality has been invoked to arrive at the last identity. By assumption \eqref{NoiseRegime}, the first term on RHS of \eqref{E_mid}  is negligible.  Thus,
\begin{align}
\mathbb{E}\left(\left|\left(\widehat{\mathbf{W}}-\mathbf{W}\right)_{mn}\right|^2\right)^{1/2}
\lesssim &
\sigma^{-1}_{mn}\sigma_\noise\sqrt{N_sN_r}.\label{Expectation}
\end{align}
This indicates that the discrepancy between the estimated and the measured scattering coefficients approaches zero very rapidly for all $m,n>K$ when $K\to+\infty$ thanks to the estimation of the magnitude of $\sigma_{mn}$. It simply means that the scattering coefficients of an inclusion $D$ can be approximated arbitrarily closely and up to \emph{any order} by the elements of $\widehat{\mathbf{W}}$ in the sense of mean-squared error when the noise regime is characterized by \eqref{NoiseRegime}. However, in view of the decay rate \eqref{Wbound} of $W^{\alpha,\beta}_{mn}$, it is reasonable to determine an adoptive resolving order $K$ by restricting the reconstruction error to be smaller than the signal level.  In particular, for any threshold reconstruction error $\varepsilon>0$, one can see from \eqref{Expectation} and \eqref{Wbound} that
\begin{align*}
\mathbb{E}\left(\left|\left(\widehat{\mathbf{W}}-\mathbf{W}\right)_{mn}\right|^2\right)^{1/2}
\lesssim &
\sigma^{-1}_{mn}\sigma_\noise\sqrt{N_sN_r}\lesssim \varepsilon\left(\frac{C_{\max}}{K}\right)^{2K-2},
\end{align*}
where $C_{\max}:=\max_{\alpha,\beta}\{C_{\alpha,\beta}\}$. After simple manipulations analogous to those in the proof of Corollary \ref{CorCond} and using the behavior of   $(H^{(1)}_{n}(\K_\alpha R))'$ for large $n$, one can show that  $\sigma_{mn}^{-1}=O((C_R^S K)^{1-K})$ for all $m,n>K$. Therefore, under noise regime characterized by  \eqref{NoiseRegime},
\begin{align*}
\left(C^S_R K\right)^{1-K} K^{2K-2} \lesssim \varepsilon \frac{ C_{\max}^{2K-2}}{\sigma_\noise},
\end{align*}
or equivalently
\begin{align*}
K^{K-1}\lesssim \varepsilon \frac{(C^2_{\max} C_R^S)^{K-1}}{\sigma_\noise}\lesssim \varepsilon \frac{h^{K-1}}{\sigma_\noise}\leq \varepsilon \snr,
\end{align*}
and the maximal resolving order is defined by
$$
K=\max\big\{N\in\mathbb{N}\;| N^{N-1}\leq \varepsilon\snr\big\}. 
$$

\section{Nearly Elastic Cloaking}\label{s:cloaking}
In this section, we consider the elastic cloaking problem as an application of the elastic scattering coefficients. The aim here is to construct an effective nearly elastic cloaking structure at a fixed frequency for making the objects inside the unit disk \emph{invisible}. We extend the approach of Ammari et al. \cite{Ammari1,Helmholtz,Maxwell} for conductivity, Helmholtz and Maxwell equations to the Lam\'e system. Towards this end, we first design
S-vanishing structures in the next subsection by \emph{canceling} the first  ESCs.

\subsection{S-vanishing Structures}\label{sec:S-vanishing}
For positive numbers $r_j$ ($j=1,2,\cdots, L+1$) with $2=r_1>r_2>\cdots>r_{L+1}=1$, construct a multi-layered structure by defining
\begin{align*}
&A_0 :=\big\{\bx\in\RR^2\,\big|\quad |\bx|>2\big\},
\\
&A_j :=\big\{\bx\in\RR^2\,\big|\quad r_{j+1}\leq|\bx| <r_j\big\}, \quad j=1,\cdots,L,
\\
&A_{L+1}:=\big\{\bx\in\RR^2\,\big|\quad|\bx|<1\big\}.
\end{align*}
Let $(\lambda_j,\mu_j,\rho_j)$ be the Lam\'e parameters and density of $A_j$, for $j=0,\cdots, L+1$. In particular, $\lambda_0$, $\mu_0$ and $\rho_0$ are the parameters of the background medium. In the sequel, the  piece-wise constant parameters  $\lambda$, $\mu$ and $\rho$ are redefined as
\begin{eqnarray}\label{S-structure}
\lambda(\bx)=\sum_{j=0}^{L+1} \lambda_j\,\chi_{(A_j)}(\bx),
\quad
\mu(\bx)=\sum_{j=0}^{L+1} \mu_j\,\chi_{(A_j)}(\bx), \quad\text{and}\quad
\rho(\bx)=\sum_{j=0}^{L+1} \rho_j\,\chi_{(A_j)}(\bx),
\end{eqnarray}
in accordance with the aforementioned multi-layered structure. The scattering coefficients $W_{m,n}^{\alpha,\beta}=W_{m,n}^{\alpha,\beta}(\lambda,\mu,\rho,\omega)$ are defined analogously as in \eqref{ESCoeff} and  $\bu^{\rm tot}=(u^{\rm tot}_1,u^{\rm tot}_2)^\top$ solves the equation
\begin{equation}\label{eq:1}
\OL_{\lambda,\mu}\bu^{\rm tot}+\rho\omega^2 \bu^{\rm tot}=0, \quad\mbox{in}\quad \R^2.
\end{equation}

Since the aforementioned multi-layered structure is circularly symmetric,  it is easy to check that the scattered field corresponding to $\bu^{\rm inc}=\bJ^{\beta}_{m}$ in $|\bx|>2$ is a linear combination of the modes $\bH_m^P$ and $\bH_m^S$.   By uniqueness of the direct scattering problems, it implies that
$$
W_{m,n}^{\alpha,\beta}=0, \qquad\quad\mbox{for all}\quad \alpha,\beta\in\{P, S\}\quad\mbox{and}\quad n\neq m.
$$
Therefore, we have the following definition of the S-vanishing structures.
\begin{defn}[S-vanishing Structure]
The medium $(\lambda,\mu,\rho)$ defined by \eqref{S-structure} is called an S-vanishing structure of order $N$ at frequency $\omega$ if $W_{n,n}^{\alpha,\beta}=0$ for all $|n|\leq N$ and $\alpha,\beta\in\{P,S\}$.
\end{defn}

In the rest of this subsection, we aim to construct an S-vanishing structure for general elastic waves. To facilitate the later analysis,  the notation $T_{\lambda,\mu}$ is adopted for the surface traction operator $\partial/\partial \nu$  associated with elastic moduli $\lambda$ and $\mu$. In order to design envisioned structure, it suffices to construct $(\lambda,\mu,\rho)$ such that $W_{n}^{\alpha,\beta}:=W_{n,n}^{\alpha,\beta}=0$ for all $0\leq n\leq N$ and $\alpha,\beta\in\{P,S\}$ thanks to Lemma \ref{Lem:Wsym1}.
We assume that the cloaked region $\{|\bx|<1\}$ is a cavity, so that the total field $\bu^{\rm tot}$ satisfies the traction-free boundary condition $T_{\lambda_{L+1},\mu_{L+1}} \bu^{\rm tot}:=\partial \bu^{\rm tot}/\partial \nu=0$ on  $|\bx|=1$. Note that the two-dimensional surface traction admits the expression
$$
T_{\lambda,\mu}\bw=2\mu(\bn\cdot \nabla w_1,\bn\cdot \nabla w_2)^\top+\lambda \Bn\, \mbox{div}\, \Bw+\mu\, \mathbf{t}\nabla_\perp\times\bw,\quad \bw=(w_1,w_2)^\top,
$$
in terms of the  surface normal and tangent vectors  $\Bn=(n_1,n_2)^\top$ and  $\mathbf{t}=(-n_1,n_2)^\top$  respectively. The solutions $\bu^{\rm tot}_n$ to \eqref{eq:1} of the form
$$
\bu^{\rm tot}_n(\bx)=\widehat{a}_j^{n,P}\bJ_n^P(\bx)+\widehat{a}_j^{n,S}\bJ_n^S(\bx)+
a_j^{n,P}\bH_n^P(\bx)+a_j^{n,S}\bH_n^S(\bx),\quad \bx\in A_j,\; j=0,\cdots,L,
$$
are sought with unknown coefficients  $\widehat{a}_j^{n,\alpha}, a_j^{n,\alpha}\in \mathbb{C}$, to be determined later. Intuitively, one should look for solutions $\bu^{\rm tot}_n$  whose coefficients fulfill the relations
\begin{align*}
\widehat{a}_0^{n,P}\;\widehat{a}_0^{n,S}\neq 0
\quad\text{and}\quad
a_0^{n,P}=a_0^{n,S}=0\quad\mbox{for all}
\quad n=0,\cdots,N.
\end{align*}
By comparison with the multipolar expansion  (\ref{u-U-exp}) of the scattered field, the scattering coefficients in this case turn out to be
\begin{eqnarray}\label{eq:5}
\begin{cases}
W_{n}^{\alpha,P}=
i4\rho_0 \omega^2 a_0^{n,\alpha}=0
& \mbox{when}\;\widehat{a}_0^{n,P}=1\,\text{ and }\,\widehat{a}_0^{n,S}=0,
\\
\nm
W_{n}^{\alpha,S}=
i4\rho_0 \omega^2 a_0^{n,\alpha}=0
& \mbox{when}\;\widehat{a}_0^{n,P}=0 \,\text{ and }\,\widehat{a}_0^{n,S}=1,
\end{cases}\quad \alpha=P,S.
\end{eqnarray}
The solution $\bu^{\rm tot}_n$ satisfies the transmission conditions
\begin{equation}\label{eq:2}
\bu^{\rm tot}_n|_+=\bu^{\rm tot}_n|_-
\quad\text{and}\quad
T_{\lambda_{j-1},\mu_{j-1}} \bu^{\rm tot}_n|_+=T_{\lambda_{j},\mu_{j}} \bu^{\rm tot}_n|_-,
\quad\text{on}\quad
|\bx|=r_j,\, j=1,\cdots, L.
\end{equation}

Fairly easily calculations indicate that on $|\bx|=r$,
\begin{align*}
\widehat{\Be}_r\cdot [T_{\lambda,\mu} \bH_n^P(\bx)]
&
= 2\mu\,\widehat{\Be}_r\cdot \frac{\partial}{\partial r}\left[\widehat{\Be}_r \frac{\partial v_n(\bx, \K_P)}{\partial r}+\frac{1}{r} \widehat{\Be}_{\theta} \frac{\partial v_n(\bx, \K_P)}{\partial\varphi_{\bx}}\right]
+ \lambda \Delta v_n(\bx, \K_P)
\\
&=
2\mu \frac{\partial ^2 v_n(\bx, \K_P)}{\partial r^2}+\lambda \Delta v_n(\bx, \K_P)
\\
&= 2\mu \K_P^2 (H_n^{(1)})''(r\K_P) \Be^{in \varphi_{\Bx}}-\lambda \K_P^2 H_n^{(1)}(r\K_P)\Be^{in
\varphi_{\Bx}}
\\
&= \frac{1}{r^2}\left( -2\mu r \K_p (H_n^{(1)})'(r \K_P)+(2\mu n^2-(\lambda+2\mu) r^2\K_p^2)H_n^{(1)}(r\K_P)    \right)\Be^{in\varphi_{\Bx}} 
\\
&=: \frac{1}{r^2} B_n^P(r\K_P,\lambda,\mu) \Be^{in\varphi_{\Bx}},
\\\nm
\widehat{\Be}_{\theta}\cdot [T_{\lambda,\mu} \bH_n^P(\bx)]
&
= 2\mu\,\widehat{\Be}_{\theta}\cdot \frac{\partial}{\partial r}\left[\widehat{\Be}_r \frac{\partial v_n(\bx, \K_P)}{\partial r}+\frac{1}{r} \widehat{\Be}_{\theta} \frac{\partial v_n(\bx, \K_P)}{\partial\varphi_{\bx}}\right]
\\
&=
2\mu \left(-\frac{1}{r^2}\frac{\partial v_n(\bx,\K_P)}{\partial \varphi_{\bx}}+\frac{1}{r}\frac{\partial^2 v_n(\bx, \K_P)}{\partial r \partial \varphi_{\bx}}  \right)
\\
&=
\frac{ (2i\mu n)}{r^2}\left( -H_n^{(1)}(r\K_P)+r\K_P  (H_n^{(1)})'(r \K_P)  \right)
\Be^{in \varphi_{\Bx}}
=:\frac{1}{r^2} C^P_n(r\K_P,\lambda,\mu) \Be^{in\varphi_{\Bx}},
\end{align*}
where
\begin{align*}
&B_n^P(t,\lambda,\mu):=-2\mu t (H_n^{(1)})'(t)+(2\mu n^2-(\lambda+2\mu) t^2)H_n^{(1)}(t),
\\
&C_n^P(t,\mu):=(2i\mu n)\left( -H_n^{(1)}(t)+t (H_n^{(1)})'(t)  \right).
\end{align*}
In the sequel, the shorthand notation 
$$
B^P_{n,j}(r)=B^P_{n}(r\K_P,\lambda_j,\mu_j)
\quad\text{and}\quad
C^P_{n,j}(r)=C^P_{n}(r\K_P,\mu_j),
$$ 
is used for simplicity. It holds that
\[
T_{\lambda_j,\mu_j} \bH_{n}^P(\bx)=\frac{1}{r_j^2}\left( B^P_{n,j}(r_j)\, \bP_n(\hat{\bx})+C_{n,j}^P(r_j) \, \bS_n(\hat{\bx})\right),\quad\mbox{on}\;|\bx|=r_j.
\]
By arguing as for $\bH_n^P$, we obtain on $|\bx|=r_j$ that
\begin{align*}
\widehat{\Be}_{r}\cdot\bH_n^S(\bx)&=2\mu \widehat{\Be}_{r}\cdot \frac{\partial}{\partial r}\left[
-\widehat{\Be}_{\theta}\frac{\partial v_n(\bx, k_s)}{\partial r}+\frac{1}{r} \widehat{\Be}_{r} \frac{\partial v_n(\bx, k_s)}{\partial \varphi_{\bx}}
\right] 
=\frac{1}{r_j^2} \,B_{n,j}^S(r_j)\, \Be^{in\varphi_{\Bx}},
\\
\widehat{\Be}_{\theta}\cdot\bH_n^S(\bx)&=
2\mu \widehat{\Be}_{\theta}\cdot \frac{\partial}{\partial r}\left[
-\widehat{\Be}_{\theta}\frac{\partial v_n(\bx, k_s)}{\partial r}+\frac{1}{r} \widehat{\Be}_{r} \frac{\partial v_n(\bx, k_s)}{\partial \varphi_{\bx}}\right] +\mu \Delta v_n(\bx, k_s)
=\frac{1}{r_j^2} \,C_{n,j}^S(r_j)\, \Be^{in\varphi_{\Bx}},
\end{align*}
where
\begin{align*}
&B_{n,j}^S(t)
:=(2i\mu_j n)
\left( -H_n^{(1)}(t)+t\left(H_n^{(1)}\right)'(t)
\right),
\\
&C_{n,j}^S(t,\mu):=2\mu_j\, t\,\left(H_n^{(1)}\right)'(t)+\left(-2\mu_j\, n^2+\mu_j\, t^2\right) H_n^{(1)}(t).
\end{align*}
Note that $B_{n,j}^S(t)=C_{n,j}^P(t)$. Therefore,
\begin{align*}
T_{\lambda_j,\mu_j} \bH_n^S(\bx)
=\widehat{\Be}_{r}\,[\widehat{\Be}_{r}\cdot\bH_n^S(\bx) ]+
\widehat{\Be}_{\theta}\,[\widehat{\Be}_{\theta}\cdot\bH_n^S(\bx) ]
=
\frac{1}{r^2_j}\Big( B^S_{n,j}(r_j)\, \bP_n(\hat{\bx})+C_{n,j}^S(r_j)\, \bS_n(\hat{\bx})\Big),
\end{align*}
with 
$$
B^S_{n,j}(r):=\,B_n^S(r\K_S,\mu_j)
\quad\text{and}\quad
C_{n,j}^S(r):=\,C_n^S(r\K_S,\mu_j).
$$
Analogously, we obtain
\[ T_{\lambda_j,\mu_j} \bJ_n^{\alpha} (\bx)=\frac{1}{r^2_j}\left( \widehat{B}^{\alpha}_{n,j}(r_j)\,\bP_n(\hat{\bx})+\widehat{C}_{n,j}^{\alpha}(r_j)\, \bS_n(\hat{\bx})\right),\quad \alpha=P,S,
\]
where $\widehat{B}^{\alpha}_{n,j}(t)$ and $\widehat{C}_{n,j}^{\alpha}(t)$ are defined in the same way as $B^{\alpha}_{n,j}(t)$ and $C^{\alpha}_{n,j}(t)$ with $H_n^{(1)}$ replaced by $J_n$.
Hence, the transmission conditions in (\ref{eq:2}) can be written as (cf. (\ref{UP})-(\ref{UStilde}))
\begin{equation}\label{eq:3}
 \mathbf{M}_{n,j-1} (r_j) (\widehat{a}_{j-1}^{n,P}, \widehat{a}_{j-1}^{n,S}, {a}_{j-1}^{n,P}, {a}_{j-1}^{n,S})^\top
= \mathbf{M}_{n,j}(r_j) (\widehat{a}_j^{n,P}, \widehat{a}_j^{n,S}, {a}_j^{n,P}, {a}_j^{n,S})^\top,
\end{equation}
for $j=1,\cdots, L$. 
Here $\mathbf{M}_{n,j}$, $j=0,\cdots,L$, $n=0,\cdots,N$, is the ${4\times 4}$ matrix defined by
\begin{eqnarray*}
\mathbf{M}_{n,j}(r):=
\begin{pmatrix}
t_{P}J^{'}_n(t_{P}) & in J_n(t_{S})           & t_{P}(H^{(1)}_n)'(t_{P}) & in H_n^{(1)}(t_{S})
\\\nm
in J_n(t_{P})         & - t_{S}J^{'}_n(t_{S}) & in H_n^{(1)}(t_{P})        & -t_{S}(H^{(1)}_n)'(t_{S})
\\\nm
\widehat{B}_{n,j}^P(t_{P}) &    \widehat{B}_{n,j}^S(t_{S}) &  {B}_{n,j}^P(t_{P}) &{B}_{n,j}^S(t_{S})
\\\nm
\widehat{C}_{n,j}^P(t_{P}) &  \widehat{C}_{n,j}^S(t_{S})   &{C}_{n,j}^P(t_{P})   &  {C}_{n,j}^S(t_{S})
\end{pmatrix}, 
\end{eqnarray*}
where $t_{\alpha}:=r\K_{\alpha}$.
The traction-free boundary condition on $|\bx|=r_{L+1}=1$ amounts to
\begin{equation}\label{eq:4}
\mathbf{M}_{n,L+1}\, (\widehat{a}_L^{n,P}, \widehat{a}_L^{n,S}, {a}_L^{n,P}, {a}_L^{n,S})^\top=(0, 0,0,0)^\top,
\qquad\quad
n=0,\cdots,N,
\end{equation}
with
\begin{eqnarray*}
\mathbf{M}_{n,L+1}:=\begin{pmatrix}
0 & 0 & 0 & 0\\
0 & 0 & 0 & 0\\
\widehat{B}_{n,L}^P(r_{L+1}\kappa_P)  & \widehat{B}_{n,L}^S(r_{L+1}\kappa_S)  & {B}_{n,L}^P(r_{L+1}\kappa_P) &{B}_{n,L}^S(r_{L+1}\kappa_S)\\
\widehat{C}_{n,L}^P(r_{L+1}\kappa_P)  &  \widehat{C}_{n,L}^S(r_{L+1}\kappa_S) &{C}_{n,L}^P(r_{L+1}\kappa_P) &  {C}_{n,L}^S(r_{L+1}\kappa_S)
\end{pmatrix}.
\end{eqnarray*}
Combining (\ref{eq:3}) and (\ref{eq:4}),  one obtains
\begin{eqnarray}\label{eq:6}
\begin{cases}
\mathbf{Q}^{(n)} (\widehat{a}_0^{n,P}, \widehat{a}_0^{n,S}, {a}_0^{n,P}, {a}_0^{n,S})^\top=(0, 0,0,0)^\top,
\\\nm
\mathbf{Q}^{(n)} =\ds\mathbf{Q}^{(n)}(\lambda,\mu,\rho\omega^2):=\mathbf{M}_{n,L+1}\,\prod_{j=1}^{L} \mathbf{M}_{n,j}^{-1}(r_j)\;\mathbf{M}_{n,j-1}(r_j)=
\begin{pmatrix}
0 & 0 \\\nm
\mathbf{Q}^{(n)}_{21} & \mathbf{Q}^{(n)}_{22}
\end{pmatrix},
\end{cases}
\end{eqnarray}
where $\mathbf{Q}^{(n)}_{21}, \mathbf{Q}^{(n)}_{22}$ are ${2\times 2}$ matrix functions of $\lambda$, $\mu$ and $\rho\omega^2$.

Exactly as in the acoustic case \cite{Helmholtz}, one can show that the determinant of $\mathbf{Q}^{(n)}_{22}$ is non-vanishing.
In fact, if $\mbox{det}(\mathbf{Q}^{(n)}_{22})=0$  then one can derive a contradiction to the uniqueness of our forward scattering problems. Therefore, it suffices to look for the parameters $\lambda_j,\mu_j,\rho_j$ ($j=1,2,\cdots,L$) from the nonlinear algebraic equations
\[
(\mathbf{Q}^{(n)}_{21})_{i,k}(\lambda,\mu,\rho\omega^2)=0,\qquad i,k=1,2,\quad n\in\NN.
\]
We are interested in a nearly S-vanishing structure of order $N$ at low frequencies, i.e.,  a structure $(\lambda,\mu, \rho)$ such that
\[
W^{\alpha,\beta}_n(\lambda,\mu, \rho,\omega)=o(\omega^{2N+2}),\quad\mbox{for all}\quad \alpha,\beta\in\{P,S\},\; |n|\leq N,\quad \mbox{as}\quad\omega\to 0.
\]
Towards this end, one needs to study the asymptotic behavior of $W^{\alpha,\beta}_n(\lambda,\mu, \rho,\omega)$ as $\omega$ tends to zero. In view of (\ref{eq:5}) and (\ref{eq:6}),  it is found  that
\begin{eqnarray}\label{eq:7}
(W_n^{\alpha,P},W_n^{\alpha,S})^\top
=i4\rho_0\omega^2\,(a_0^{n,P},a_0^{n,S})^\top
=-i4\rho_0 \omega^2 \,(\mathbf{Q}_{22}^{(n)})^{-1} \;\mathbf{Q}_{21}^{(n)}\,(\widehat{a}_0^{n,P},\widehat{a}_0^{n,S})^\top,
\end{eqnarray}
where $\widehat{a}_0^{n,P}$ and $\widehat{a}_0^{n,S}$ are selected depending on (\ref{eq:5}).

Let $\bW_n$ denote the $2\times 2$ matrix
$$
\bW_{n}=
\begin{pmatrix}
W^{P,P}_{n} & W^{S,P}_{n}
\\\nm
W^{P,S}_{n} & W^{S,S}_{n}
\end{pmatrix}.
$$
Then, the following result based on relation \eqref{eq:7} elucidates the low frequency asymptotic behavior of $\bW_n$.

\begin{thm}\label{theorem-cloaking}
For all $n\in\NN$, we have
\begin{align}
\mathbf{W}_n(\lambda,\mu,\rho,\omega)=\omega^{2n+2}\left(
\bV_{n,0}(\lambda,\mu,\rho)+\sum_{l=0}^{N-n}\sum_{j=0}^{(L+1)l} \omega^{2l}\,(\ln \omega)^j\, \bV_{n,l,j}(\lambda,\mu,\rho)\right)+\mathbf{\Upsilon}_n,
\label{eq:8}
\end{align}
as $\omega\to 0$. Here matrices $\bV_{n,0}$ and $\bV_{n,l,j}$ are defined by
$$
\bV_{n,0}=\begin{pmatrix}
V^{P,P}_{n,0} & V^{S,P}_{n,0}
\\\nm
V^{P,S}_{n,0} & V^{S,S}_{n,0}
\end{pmatrix}
\quad\text{and}\quad
\bV_{n,l,j}=\begin{pmatrix}
V^{P,P}_{n,l,j} & V^{S,P}_{n,l,j}
\\\nm
V^{P,S}_{n,l,j} & V^{S,S}_{n,l,j}
\end{pmatrix},
$$
in terms of some $V^{\alpha,\beta}_{n,0}$ and $V^{\alpha,\beta}_{n,l,j}$ dependent on $\lambda,\mu,\rho$ but independent of $\omega$. The residual matrix $\mathbf{\Upsilon}_n=(\Upsilon^n_{ik})_{i,k=1,2}$ is such that $|\Upsilon^n_{ik}|\leq C\omega^{2N+2}$, for all $i,k=1,2$, where constant $C\in\RR_+$ is independent of $\omega$.
\end{thm}

The analytic expressions of the quantities $V_{n,0}^{\alpha,\beta}$ and $V_{n,l,j}^{\alpha,\beta}$ in terms of $\lambda_j$, $\mu_j$ and $\rho_j$ are very complicated, but can be extracted, for example, by using the symbolic toolbox of MATLAB.
Theorem \ref{theorem-cloaking} follows from (\ref{eq:7}) and the low-frequency asymptotics of $\mathbf{Q}_{22}^{(n)}(\lambda,\mu,\rho\omega^2)$ and $\mathbf{Q}_{21}^{(n)}(\lambda,\mu,\rho\omega^2)$ as $\omega\to 0$. The latter can be derived based on the definition given in (\ref{eq:6}) in combination with the expansion formula of Bessel and Neumann functions, and their derivatives for small arguments.
For the sake of completeness, we sketch the proof of Theorem \ref{theorem-cloaking} in the Appendix \ref{AppendD}.

In order to construct a nearly S-vanishing structure of order $N$ at low frequencies, thanks to Theorem \ref{theorem-cloaking}, one needs to determine the parameters $\lambda_j,\mu_j$ and $\rho_j$ from the equations
$$
V_{n,0}^{\alpha,\beta}(\lambda,\mu,\rho)=V_{n,l,j}^{\alpha,\beta}(\lambda,\mu,\rho)=0,
$$
for all $0\leq n\leq N$, $1\leq l\leq (N-n)$, $1\leq j\leq (L+1)l$ and $\alpha,\beta\in\{P,S\}$.
It should be emphasized that one does not know if a solution does exist for any order $N$.
Numerically, this can be achieved by applying, e.g., the gradient descent method to the minimization problem
\begin{eqnarray}\label{minimization}
\min_{\lambda_j,\mu_j,\rho_j} \sum_{\alpha,\beta\in\{P,S\}} \left\{\left|V_{n,0}^{\alpha,\beta}\right|^2+ \sum_{l=0}^{N-n}\sum_{j=0}^{(L+1)l} \,\left|V_{n,l,j}^{\alpha,\beta}\right|^2\right\}.
\end{eqnarray}

\subsection{Enhancement of Near-Cloaking in Elasticity}

The aim of this section is to show that the nearly S-vanishing structures constructed in Section \ref{sec:S-vanishing} can be used to enhance cloaking effect in elasticity. The enhancement of near-cloaking is based on the idea of transformation
optics (also called the scheme of changing variables) used in \cite{GLU,GLU2,Parnell,PenSchSmi}. Let $(\lambda,\mu,\rho)$ be a nearly S-vanishing structure of order $N$ of the form of (\ref{S-structure}) at low frequencies.
It implies that for some fixed $\omega>0$ there exists $\epsilon_0>0$ such that
\[
\left|W_{m,n}^{\alpha,\beta}[\lambda,\mu,\rho,\epsilon\omega]\right|=o(\epsilon^{2N+2}),\quad |n|\leq N,\quad \epsilon\leq \epsilon_0.
\]
On the other hand, using the asymptotic behavior of $\bJ^{\beta}_{n}(\bx,\K_\alpha)$ as $\omega\rightarrow 0$ for $\alpha,\beta=P,S$, one can derive  
 from the proof of Lemma \ref{LemWbound} that
\begin{eqnarray*} 
\left|W_{n}^{\alpha,\beta}[\lambda,\mu,\rho, \epsilon \omega]\right|\leq C\,\epsilon^{2N},\quad \forall\, |n|\geq N,\,\,\epsilon\leq \epsilon_0.
\end{eqnarray*}
Hence, by Theorem \ref{ThmFarAmp}, the far-field elastic scattering amplitudes can be estimated by
\begin{eqnarray*}
\bu^{\infty}_{\alpha}[\lambda,\mu,\rho,\epsilon \omega]({\hbx},{\hbx}')=o(\epsilon^{2N-1}),\quad \alpha=P,S,\quad\mbox{as}\quad \epsilon\rightarrow 0,
\end{eqnarray*}
uniformly in all observation directions $\hat{\bx}$ and incident directions $\hat{\bx}'$, noting that
\begin{eqnarray*}
\gamma_n^\alpha\sim O(\epsilon^{2N-2}),\quad A_n^{\infty,\alpha}\sim O(\epsilon),\quad\mbox{as}\quad \epsilon\rightarrow0,\quad \alpha=P,S.
\end{eqnarray*}

Introduce the transformation $\Psi_\epsilon: \R^2\to \R^2$ by
\[
\Psi_\epsilon(\bx):=\frac{1}{\epsilon}\bx,\quad \bx\in \R^2.
\]
Then, by arguing as in the acoustic and electromagnetic case \cite{Helmholtz,Maxwell},  we have
\[
\bu^{\infty}_{\alpha}[\lambda\circ\Psi_\epsilon ,\mu\circ\Psi_\epsilon,\rho\circ\Psi_\epsilon,\omega]
=\bu^{\infty}_{\alpha}[\lambda,\mu,\rho,\epsilon \omega]
=o(\epsilon^{2N-1}),\quad\mbox{for all}\quad \epsilon\leq \epsilon_0.
\]
Note that the medium $(\lambda\circ\Psi_\epsilon ,\mu\circ\Psi_\epsilon,\rho\circ\Psi_\epsilon)$ is a homogeneous multi-coated structure of radius $2\epsilon$.

We now apply the transformation invariance of the Lam\'e system to the medium $(\lambda\circ\Psi_\epsilon ,\mu\circ\Psi_\epsilon,\rho\circ\Psi_\epsilon)$.
Recall that
the elastic wave propagation in such a homogeneous isotropic medium can be restated as
\begin{eqnarray}
\label{Lame}
\nabla\cdot (\mathfrak{C}: \nabla \bu^{\rm tot})+\omega^2(\rho\circ\Psi_\epsilon) \bu^{\rm tot}=0,\qquad\mbox{in}\,\,\R^2,
\end{eqnarray}
where $\mathfrak{C}=(C_{ijkl})_{i,j,k,l=1,\cdots, N}$  is the rank-four stiffness tensor defined by
\begin{eqnarray}
\label{C}
 C_{ijkl}(\bx)=(\lambda\circ\Psi_\epsilon)\, \delta_{i,j}\delta_{k,l}+(\mu\circ\Psi_\epsilon)\, ( \delta_{i,k}\delta_{j,l}+ \delta_{i,l}\delta_{j,k}),
\end{eqnarray} and the action of $\mathfrak{C}$ on a matrix $\mathbf{A}=(a_{ij})_{i,j=1,2}$ is defined as
$$
\mathfrak{C}:\mathbf{A}=(\mathfrak{C}:\mathbf{A})_{i,j=1,2} =\left(\sum_{k,l=1,2} C_{ijkl}\; a_{kl}\right)_{i,j=1,2}.
$$
In the case of a generic anisotropic elastic material, the stiffness tensor satisfies the   symmetries
\begin{equation*} 
\quad C_{ijkl}=C_{klij}
\,\, \text{(major symmetry)} \qquad\text{and}\qquad
 C_{ijkl}=C_{jikl}=C_{ijlk}
 \,\,\text{(minor symmetry)},
\end{equation*}
for all $i,j,k,l=1,2$.
Let $\tilde{\bx}=(\tilde{x}_1,\tilde{x}_2)=F_\epsilon(\bx): \R^2\rightarrow \R^2$ be a bi-Lipschitz and orientation-preserving transformation such that
$F_{\epsilon}(\{|\bx|<\epsilon\})=\{|\tilde{\bx}|<1\}$ and that the region $|\bx|\geq2$ remains invariant under the transformation.
This implies that we have blown up a small traction-free disk of radius $\epsilon<1$ to the unit disk centered at the origin.
The push-forwards of $\mathfrak{C}$ and $\rho$ are defined respectively by
\begin{align*}\label{eq:pushforward}
(F_\epsilon)_{*}\mathfrak{C}:=&\widehat{\mathfrak{C}}=\left(\widehat{C}_{iqkp}(\tilde{x})\right)_{i,q,k,p=1,2} =\left(\frac{1}{\det(\mathbf{M})}\left\{\displaystyle\sum_{l,j=1,2}  C_{ijkl}  \frac{\partial \tilde{x}_p}{\partial x_l}
\frac{\partial \tilde{x}_q}{\partial x_j}\right\}\bigg |_{x=F_\epsilon^{-1}(\tilde{x})}\right)_{i,q,k,p=1,2},
\\
(F_\epsilon)_{*}\rho:=&\widehat{\rho}= \left(\frac{\rho}{\det(\mathbf{M})}\right)\bigg |_{x=F_\epsilon^{-1}(\tilde{x})},\quad \mathbf{M}=\left(\frac{\partial \tilde{x}_i}{\partial x_j}\right)_{i,j=1,2}.
\end{align*}
 We need the following lemma (see, for instance, \cite{HL2015, MBW}).
 \begin{lem}
 The function $\bu^{\rm tot}$ is a solution to 
 $$
 \nabla\cdot (\mathfrak{C}: \nabla \bu^{\rm tot})+\omega^2\rho \bu^{\rm tot}=0, 
 \qquad\text{in}\,\, \R^2,
 $$  
 if and only if $\widehat{\bu}^{\rm tot}=\bu^{\rm tot}\circ (F_\epsilon)^{-1}$ satisfies 
 $$
 \widehat{\nabla}\cdot (\widehat{\mathfrak{C}}: \widehat{\nabla} \widehat{\bu}^{\rm tot})+\omega^2\widehat{\rho} \widehat{\bu}^{\rm tot}=0,
 \qquad\text{in}\,\, \R^2,
 $$ 
 where $\widehat{\nabla}$ denotes the gradient operator with respect to transformed variable $\tilde{\bx}$.
 \end{lem}
Applying the above lemma to the Lam\'e system (\ref{Lame}),  one obtains the following result.
\begin{thm}\label{Them}
If $(\lambda,\mu,\rho)$ is a nearly S-vanishing structure of order $N$ at low frequencies then there exists $\epsilon_0>0$ such that
\[
\bu^{\infty}_{\alpha}[(F_\epsilon)_{*}\mathfrak{C},(F_\epsilon)_{*}(\rho\circ\Psi_\epsilon ),\omega](\bx,\bx')
=o(\epsilon^{2N-1}),\qquad \alpha=P,S,
\]
for all $\epsilon<\epsilon_0$, uniformly in all $\bx$ and $\bx'$. Here the stiffness tensor $\mathfrak{C}$ is defined by (\ref{C}).   Moreover, the elastic medium $((F_\epsilon)_{*}\mathfrak{C},(F_\epsilon)_{*}(\rho\circ\Psi_\epsilon ))$ in $1<|\bx|<2$ is a nearly cloaking device for the \emph{hidden} region $|\bx|<1$.
\end{thm}

Theorem \ref{Them} implies that  for any frequency $\omega$ and any integer number $N$ there exist $\epsilon_0=\epsilon_0(\omega,N)>0$ and
the elastic medium $((F_\epsilon)_{*}\mathfrak{C},(F_\epsilon)_{*}(\rho\circ\Psi_\epsilon ))$ with $\epsilon<\epsilon_0$ such that the nearly cloaking enhancement can be achieved at the order $o(\epsilon^{2N-1})$.

We finish this section with the following remarks.
\begin{rem}
 Unlike the acoustic and electromagnetic case, the transformed elastic tensor $(F_\epsilon)_{*}\mathfrak{C}$ is not anisotropic since it possesses the major symmetry only. Note that the transformed mass density $(F_\epsilon)_{*}(\rho\circ\Psi_\epsilon )$ is still isotropic. In fact, it has been pointed out by Milton, Briane, and Willis \cite{MBW} that the invariance of the Lam\'e system can be achieved only if one relaxes the assumption of the minor symmetry of the transformed elastic tensor. This has led Norris and Shuvalov \cite{Norris11} and Parnell \cite{Parnell} to explore the elastic cloaking by using Cosserat material or by employing non-linear pre-stress in a neo-Hookean elastomeric material.\end{rem}

In this section, we have designed an enhanced nearly cloaking device for general incoming elastic plane waves. A device for cloaking only compression or shear waves can be analogously constructed by using the corresponding elastic scattering coefficients. Note that
the medium $(\lambda,\mu,\rho)$ defined by \eqref{S-structure} is called an S-vanishing structure of order $N$ for compression (resp. shear) waves if $W_{n,n}^{\alpha,P}=0$ (resp. $W_{n,n}^{\alpha,S}=0$ ) for all $|n|\leq N$ and $\alpha\in\{P,S\}$. In this case, one needs to seek parameters $\lambda_j, \mu_j$ and $\rho_j$ for solving the minimization problem (\ref{minimization}) with $\beta=P$ (resp. $\beta=S$). The estimate in Theorem \ref{Them} can be analogously achieved for nearly cloaking compression or shear waves.

\section{Discussion}\label{s:conc}

In this article, the elastic scattering coefficients (ESC) of characteristically small inclusions are discussed using  surface vector harmonics based cylindrical solutions to Lam\'e equations and the multipolar expansions of elastic fields based on them. An integral equations based approach is used. It is established that the scattered field and the far field scattering amplitudes admit natural expansions in terms of ESC. This connection substantiates their utility in direct and inverse elastic scattering. The scattering coefficients of a three-dimensional elastic inclusion can be analogously defined using three-dimensional vector spherical harmonics and specially constructed vector wave functions. An added complication in three dimensions is that there are three wave-modes (P, SV and SH modes) which cannot be completely decoupled. It can be easily proved that the ESC possess similar properties in three dimensions.

The decay rates and the symmetry properties of the ESC  are also discussed. The symmetry of the ESC can be traced back to reciprocity property of scattered waves in elastic media. These properties also indicate that only first few coefficients are significant and sufficient to cater to a variety of scattering problems. The high-order ESC contain fine details of shape oscillations and geometric features of the inclusion. Thus, the largest order of stably recoverable ESC determines the maximal resolving power of the imaging setup and determines the resolution limit in feature extraction frameworks.

For reconstructing significant ESC from multi-static response data, we have formulated a truncated linear system of equations where the truncation parameter can be tuned depending on the requirements of the actual physical problem, stability constraints, truncation error, and the measurement noise. This truncated system is converted to a matrix system wherein all the ESC up to truncation order are arranged into a matrix that happens to be Hermitian. The system is shown to be ill-conditioned, however, the additional constraints dictated by the reciprocity principle and conservation of energy can be effectively used to gain stability.
The Hermitian property is pertinent to designing subspace migration type shape identification frameworks in elastic media. Moreover, shape descriptors and invariant features of elastic objects can also be discussed using ESC. This will be further discussed in a forthcoming article.

As an application of ESC, we constructed the scattering coefficients vanishing structures and elucidated that such structures can be used to enhance the performance of nearly elastic cloaking devices. The designed nearly cloaking scheme can be used for cloaking only compression or shear waves by using the corresponding ESC. This
suggests that ESC proposed in this article play more important role 
than the acoustic scattering coefficients. The results presented in the article are not restricted to only two dimensions and can be easily extended to three dimensions.

In future studies, the role of ESC in mathematical imaging, especially from the perspectives of designing shape invariant and descriptors in elastic media, will be investigated. Moreover, in order to handle inverse elastic medium  scattering problems and to understanding the super-resolution phenomena in elastic media, the concept of heterogeneous ESC will be discussed.

\appendix
\section{Multipolar Expansion of Elastodynamic Fundamental Solution}\label{AppendA}

Note that, by Helmholtz decomposition, $\bGamma^\omega(\bx,\by)\bp$  can be decomposed for any constant vector $\bp\in\RR^2$ and $\bx\neq \by$ as (see, for instance, \cite{Princeton, Kupradze79})
\begin{align}\label{Decomp}
\bGamma^\omega(\bx,\by)\bp=\nabla_\bx G_P (\bx,\by) +\vec{\nabla}_{\perp,\bx}\times G_S (\bx,\by),
\end{align}
with
$$
G_P(\bx,\by):=-\frac{1}{\K_P^2}\nabla_\bx \cdot (\bGamma^\omega(\bx,\by)\bp)
\quad\text{and}\quad
G_S(\bx,\by):=\frac{1}{\K_S^2} {\nabla}_{\perp,\bx}\times(\bGamma^\omega(\bx,\by)\bp).
$$
By \eqref{Green_fun}, one can easily show that
\begin{align}\label{bGP}
&
G_P(\bx,\by)=-\frac{1}{\rho_0\omega^2} \nabla_{\bx}g(\bx-\by,\K_P)\cdot\bp
=\frac{1}{\rho_0\omega^2} \nabla_{\by}g(\bx-\by,\K_P)\cdot\bp,\\ \label{bGS}
&
G_S(\bx,\by)=-\frac{1}{\rho_0\omega^2} \vec{\nabla}_{\perp,\bx}\times g(\bx-\by,\K_S)\cdot\bp
=\frac{1}{\rho_0\omega^2} \vec{\nabla}_{\perp,\by}\times g(\bx-\by,\K_S)\cdot\bp,
\end{align}
where the reciprocity relations
$$
g(\bx-\by,\K_\alpha)= g(\by-\bx,\K_\alpha)
\quad\text{and}\quad
\nabla_\bx g(\bx-\by,\K_\alpha)= -\nabla_\by g(\bx-\by,\K_\alpha),
$$
have been used. Recall that, by Graf's addition formula (see, for example, \cite[Formula 10.23.7]{NIST}), we have
$$
H_0^{(1)}(\K|\bx-\by|)=\sum_{n\in } H_n^{(1)}(\K|\bx|)e^{in\theta_\bx} \overline{J_n(\K|\by|)e^{in\theta_\by}}.
$$
Consequently, it follows from \eqref{bGP},\eqref{bGS}, and \eqref{green-fn-g} that
\begin{align*}
G_P(\bx,\by)&=\frac{i}{4\rho_0\omega^2}\sum_{n\in } H_n^{(1)}(\K_P|\bx|)e^{in\theta_\bx}\, \overline{\nabla [ J_n(\K_P|\by|)e^{in\theta_\by}]\cdot\bp},\\
G_S(\bx,\by)&=\frac{i}{4\rho_0\omega^2}\sum_{n\in } H_n^{(1)}(\K_S|\bx|)e^{in\theta_\bx} \overline{\vec{\nabla}_{\perp}\times[ J_n(\K_S|\by|)e^{in\theta_\by}]\cdot\bp}.
\end{align*}
The identity \eqref{G-multipole} follows by substituting the values of $G_\alpha$ in the decomposition \eqref{Decomp} and using the definition of $\bJ^\alpha$ and $\bH^\alpha$.

\section{Proof of Lemma \ref{Lem:Wsym2}}\label{AppendB}
\begin{proof}
Since our formulation here is based on an integral representation in terms of the densities $\bphi$ and $\bpsi$ satisfying \eqref{integral-system}, we take a different route than those already discussed in \cite{Varath, Waterman} without directly invoking the argument of reciprocity.

Let us first fix some notation. For any $\bv,\bw\in H^{3/2}(D)^2$  and $a,b\in\RR_+$,   define the quadratic form
\begin{eqnarray*}
\langle \bv,\bw\rangle^{a,b}_D:=\int_D\left[a(\nabla\cdot\bv)(\nabla\cdot\bw)+\frac{b}{2}\left(\nabla\bv+\nabla\bv^\top\right):(\nabla\bw+\nabla\bw^\top)\right]d\bx, 
\end{eqnarray*}
where double dot ` $ : $ ' denotes the matrix contraction operator defined for two matrices $\mathbf{A}=(a_{ij})$ and $\mathbf{B}=(b_{ij})$ by $\mathbf{A}:\mathbf{B}:=\ds\sum_{i,j}a_{ij}b_{ij}$. It is easy to get from the definition of $\langle\cdot,\cdot\rangle^{a,b}_D$ that
\begin{align}
\int_{\partial D}\bv\cdot\frac{\partial\bw}{\partial\nu}d\sigma(\bx)=\int_{D}\bv\cdot\OL_{a,b}[\bw]d\bx+\langle \bv,\bw\rangle^{a,b}_D.\label{green1}
\end{align}
Note that if $\bw$ is a solution of the Lam\'e equation  $\OL_{a,b}[\bw]+c\omega^2\bw=\mathbf{0}$ then
\begin{align*}
\int_{\partial D}\bv\cdot\frac{\partial\bw}{\partial\nu}d\sigma(\bx)=-c \omega^2\int_{D}\bv\cdot\bw d\bx+\langle \bv,\bw\rangle^{a,b}_D, 
\end{align*}
and consequently from \eqref{green1}
\begin{align}
\int_{\partial D}\bv\cdot\frac{\partial\bw}{\partial\nu}d\sigma(\bx)=\int_{\partial D} \frac{\partial\bv}{\partial\nu}\cdot\bw d\sigma(\bx)-c\omega^2\int_{D}\bv\cdot\bw d\bx-\int_{D}\OL_{a,b}[\bv]\cdot\bw d\bx.\label{green3}
\end{align}
Moreover, if $\bv$ solves $\OL_{a,b}[\bv]+c\omega^2\bv=\mathbf{0}$  then
\begin{align}
\label{green4}
\int_{\partial D}\bv\cdot\frac{\partial\bw}{\partial\nu}d\sigma(\bx)=\int_{\partial D}\frac{\partial\bv}{\partial\nu}\cdot \bw d\sigma(\bx).
\end{align}
We will also require the constants
\begin{align*}
\eta_P:=&\frac{\mu_0}{\mu_1-\mu_0},
\\
\widetilde{\eta}_P:=&\frac{\mu_1}{\mu_1-\mu_0},
\\
\eta_S:=&\frac{\lambda_0+\mu_0}{(\lambda_1-\lambda_0)+(\mu_1-\mu_0)},
\\
\widetilde{\eta}_S:=&\frac{\lambda_1+\mu_1}{(\lambda_1-\lambda_0)+(\mu_1-\mu_0)}.
\end{align*}

Let $(\bphi_n^\alpha,\bpsi_n^\alpha)$ and $(\bphi_n^\beta,\bpsi_m^\beta)$ be the solutions of $\eqref{integral-system}$ with $\bu^{\rm inc}=\bJ^\alpha$ and $\bu^{\rm inc}=\bJ^\beta$ respectively, i.e.
\begin{align}
&\widetilde{\mathcal{S}}^\omega_D\bphi_n^\alpha-\mathcal{S}^\omega_D\bpsi_n^\alpha
=\bJ^\alpha_n\big|_{\partial D},\label{1}
\\
&\frac{\partial}{\partial\widetilde{\nu}}\widetilde{\mathcal{S}}^\omega_D\bphi_n^\alpha\Big|_-
-\frac{\partial}{\partial\nu}\mathcal{S}^\omega_D\bpsi_n^\alpha\Big|_+
=\ds\frac{\partial\bJ^\alpha_n}{\partial\nu} \Big|_{\partial D},\label{2}
\end{align}
and
\begin{align}
&\widetilde{\mathcal{S}}^\omega_D\bphi_m^\beta-\mathcal{S}^\omega_D\bpsi_m^\beta=\bJ^\beta_m\big|_{\partial D},\label{3}
\\
&\frac{\partial}{\partial\widetilde{\nu}}\widetilde{\mathcal{S}}^\omega_D\bphi_m^\beta\Big|_-
-\frac{\partial}{\partial\nu}\mathcal{S}^\omega_D\bpsi_m^\beta\Big|_+=\ds\frac{\partial\bJ^\beta_m}{\partial\nu} \Big|_{\partial D}.\label{4}
\end{align}
Then, by making use of the jump conditions \eqref{Sjumps},  $W^{\alpha,\beta}_{m,n}$ can be expressed as
\begin{align*}
W^{\alpha,\beta}_{m,n}=\ds\int_{\partial D}\overline{\bJ_n^\alpha}\cdot\bpsi_m^\beta d\sigma(\bx)= \int_{\partial D}\overline{\bJ_n^\alpha}\cdot\left[\frac{\partial}{\partial\nu}\mathcal{S}^\omega_D[\bpsi^\beta_m]\Big|_+-\frac{\partial}{\partial\nu}\mathcal{S}^\omega_D[\bpsi^\beta_m]\Big|_-\right] d\sigma(\bx).
\end{align*}
Further, by invoking \eqref{4} and subsequently using \eqref{green3} and \eqref{green4}, one gets the expression
\begin{align*}
W^{\alpha,\beta}_{m,n}=&-\int_{\partial D}\overline{\bJ_n^\alpha}\cdot\frac{\partial\bJ^\beta_m}{\partial\nu}d\sigma(\bx)+\int_{\partial D}\overline{\bJ_n^\alpha}\cdot\left[\frac{\partial}{\partial\widetilde{\nu}}\widetilde{\mathcal{S}}^\omega_D[\bphi^\beta_m]\Big|_-
-\frac{\partial}{\partial\nu}\mathcal{S}^\omega_D[\bpsi^\beta_m]\Big|_-\right] d\sigma(\bx)
\\
\nm
=&-\int_{\partial D}\overline{\bJ_n^\alpha}\cdot\frac{\partial\bJ^\beta_m}{\partial\nu}d\sigma(\bx)+\int_{\partial D}\left[\frac{\partial\overline{\bJ_n^\alpha}}{\partial\widetilde{\nu}}\cdot\widetilde{\mathcal{S}}^\omega_D[\bphi^\beta_m]
-\frac{\partial\overline{\bJ_n^\alpha}}{\partial\nu}\cdot\mathcal{S}^\omega_D[\bpsi^\beta_m]\right] d\sigma(\bx)
\\
\nm
&-\rho_1\omega^2\int_{D}\overline{\bJ^\alpha_n}\cdot\widetilde{\mathcal{S}}^\omega_D[\bphi^\beta_m]d\bx-\int_{D}\OL_{\lambda_1,\mu_1}[\overline{\bJ^\alpha_n}]\cdot\widetilde{\mathcal{S}}^\omega_D[\bphi^\beta_m]d\bx.
\end{align*}
This, together with \eqref{3}, leads to
\begin{align*}
W^{\alpha,\beta}_{m,n}=&-\int_{\partial D}\overline{\bJ_n^\alpha}\cdot\frac{\partial\bJ^\beta_m}{\partial\nu}d\sigma(\bx)
+
\int_{\partial D}\frac{\partial\overline{\bJ_n^\alpha}}{\partial\widetilde{\nu}}\cdot\widetilde{\mathcal{S}}^\omega_D[\bphi^\beta_m] d\sigma(\bx)
-\int_{\partial D}\frac{\partial\overline{\bJ_n^\alpha}}{\partial\nu}\cdot\widetilde{\mathcal{S}}^\omega_D[\bphi^\beta_m]d\sigma(\bx)
\\
&+\int_{\partial D}\frac{\partial\overline{\bJ}_n^\alpha}{\partial\nu}\cdot\bJ_m^\beta\cdot d\sigma(\bx)
-\rho_1\omega^2\int_{D}\overline{\bJ^\alpha_n}\cdot\widetilde{\mathcal{S}}^\omega_D[\bphi^\beta_m]d\bx-\int_{D}\OL_{\lambda_1,\mu_1}[\overline{\bJ^\alpha_n}]\cdot\widetilde{\mathcal{S}}^\omega_D[\bphi^\beta_m]d\bx.
\end{align*}
It is easy to see that the first and the fourth terms cancel out each other thanks to \eqref{green4}. Therefore,
\begin{align}
W^{\alpha,\beta}_{m,n}=&
\int_{\partial D}\left[\frac{\partial\overline{\bJ_n^\alpha}}{\partial\widetilde{\nu}}-\frac{\partial\overline{\bJ_n^\alpha}}{\partial\nu}\right]\cdot\widetilde{\mathcal{S}}^\omega_D[\bphi^\beta_m]d\sigma(\bx)-\rho_1\omega^2\int_{D}\overline{\bJ^\alpha_n}\cdot\widetilde{\mathcal{S}}^\omega_D[\bphi^\beta_m]d\bx
\nonumber
\\
&-\int_{D}\OL_{\lambda_1,\mu_1}[\overline{\bJ^\alpha_n}]\cdot\widetilde{\mathcal{S}}^\omega_D[\bphi^\beta_m]d\bx.\label{main1}
\end{align}

Remark that $\nabla\cdot\bJ^S_n=0=\nabla\times\bJ^P_n$. Therefore, it is easy to verify  by definition of the surface traction operator  that
\begin{align}
\frac{\partial\overline{\bJ_n^\alpha}}{\partial\widetilde{\nu}}-\frac{\partial\overline{\bJ_n^\alpha}}{\partial\nu}=\frac{1}{\eta_\alpha}\frac{\partial\overline{\bJ_n^\alpha}}{\partial\nu} =\frac{1}{\widetilde{\eta}_\alpha}\frac{\partial\overline{\bJ_n^\alpha}}{\partial\widetilde{\nu}}.\label{main2}
\end{align}
Thus, using right most quantity of \eqref{main2} in \eqref{main1} and subsequently invoking identity \eqref{green3}, one gets
\begin{align*}
\widetilde{\eta}_\alpha W^{\alpha,\beta}_{m,n}
=& \int_{\partial D}\frac{\partial\overline{\bJ^\alpha_n}}{\partial\widetilde{\nu}}\cdot\widetilde{\mathcal{S}}^\omega_D[\bphi^\beta_m]d\sigma(\bx)-\widetilde{\eta}_\alpha\rho_1\omega^2\int_{D}\overline{\bJ^\alpha_n}\cdot\widetilde{\mathcal{S}}^\omega_D[\bphi^\beta_m]d\bx
\\
&-\widetilde{\eta}_\alpha\int_{D}\OL_{\lambda_1,\mu_1} [\overline{\bJ^\alpha_n} ]\cdot\widetilde{\mathcal{S}}^\omega_D[\bphi^\beta_m]d\bx
\\
=&\int_{\partial D}\overline{\bJ^\alpha_n}\cdot\frac{\partial}{\partial\widetilde{\nu}}\widetilde{\mathcal{S}}^\omega_D[\bphi^\beta_m]\Big|_-d\sigma(\bx)
+(1-\widetilde{\eta}_\alpha)\rho_1\omega^2\int_{D}\overline{\bJ^\alpha_n}\cdot\widetilde{\mathcal{S}}^\omega_D[\bphi^\beta_m]d\bx
\\
&+(1-\widetilde{\eta}_\alpha)\int_{D}\OL_{\lambda_1,\mu_1} [\overline{\bJ^\alpha_n}]\cdot\widetilde{\mathcal{S}}^\omega_D[\bphi^\beta_m]d\bx.
\end{align*}
This, together with \eqref{1} and \eqref{4}, provides
\begin{align}
\widetilde{\eta}_\alpha W^{\alpha,\beta}_{m,n}
=&
\int_{\partial D}\overline{\widetilde{\mathcal{S}}^\omega_D[\bphi^\alpha_n]}\cdot\frac{\partial}{\partial\widetilde{\nu}}\widetilde{\mathcal{S}}^\omega_D[\bphi^\beta_m]\Big|_-d\sigma(\bx)
-\int_{\partial D}\overline{\mathcal{S}^\omega_D[\bpsi^\alpha_n]}\cdot\frac{\partial}{\partial\nu}\mathcal{S}^\omega_D[\bpsi^\beta_m]\Big|_+d\sigma(\bx)
\nonumber
\\
&-\int_{\partial D}\overline{\mathcal{S}_D^\omega[\bpsi^\alpha_n]}\cdot\frac{\partial\bJ^\beta_m}{\partial\nu}d\sigma(\bx)+(1-\widetilde{\eta}_\alpha)\rho_1\omega^2\int_{D}\overline{\bJ^\alpha_n}\cdot\widetilde{\mathcal{S}}^\omega_D[\bphi^\beta_m]d\bx
\nonumber
\\
&+(1-\widetilde{\eta}_\alpha)\int_{D}\OL_{\lambda_1,\mu_1}[\overline{\bJ^\alpha_n}]\cdot\widetilde{\mathcal{S}}^\omega_D[\bphi^\beta_m]d\bx.\label{main3}
\end{align}
Similarly, substituting the first relation of \eqref{main2} back in \eqref{main1} and invoking \eqref{3}, one obtains
\begin{align}
\eta_\alpha W^{\alpha,\beta}_{m,n}=&\int_{\partial D}\frac{\partial\overline{\bJ^\alpha_n}}{\partial\nu}\cdot\widetilde{\mathcal{S}}^\omega_D[\bphi^\beta_m]d\sigma(\bx)-\eta_\alpha\rho_1\omega^2\int_{D}\overline{\bJ^\alpha_n}\cdot\widetilde{\mathcal{S}}^\omega_D[\bphi^\beta_m]d\bx
\nonumber
\\
&-\eta_\alpha\int_{D}\OL_{\lambda_1,\mu_1}[\overline{\bJ^\alpha_n}]\cdot\widetilde{\mathcal{S}}^\omega_D[\bphi^\beta_m]d\bx
\nonumber
\\
=&
\int_{\partial D}\frac{\partial\overline{\bJ^\alpha_n}}{\partial\nu}\cdot\mathcal{S}^\omega_D[\bpsi^\beta_m]d\sigma(\bx)+\int_{\partial D}\frac{\partial\overline{\bJ^\alpha_n}}{\partial\nu}\cdot\bJ^\beta_m d\sigma(\bx)
\nonumber
\\
&-\eta_\alpha\rho_1\omega^2\int_{D}\overline{\bJ^\alpha_n}\cdot\widetilde{\mathcal{S}}^\omega_D[\bphi^\beta_m]d\bx
-\eta_\alpha\int_{D}\OL_{\lambda_1,\mu_1}[\overline{\bJ^\alpha_n}]\cdot\widetilde{\mathcal{S}}^\omega_D[\bphi^\beta_m]d\bx.\label{main4}
\end{align}
Finally, subtracting \eqref{main4} from \eqref{main3} and noting that $\widetilde\eta_\alpha-\eta_\alpha=1$, one finds out that
\begin{align}
W^{\alpha,\beta}_{m,n}=&
\int_{\partial D}\overline{\widetilde{\mathcal{S}}^\omega[\bphi^\alpha_n]}\cdot\frac{\partial}{\partial\widetilde{\nu}}\widetilde{\mathcal{S}}^\omega_D[\bphi^\beta_m]\Big|_-d\sigma(\bx)
-\int_{\partial D}\overline{\mathcal{S}^\omega_D[\bpsi^\alpha_n]}\cdot\frac{\partial}{\partial\nu}\mathcal{S}^\omega_D[\bpsi^\beta_m]\Big|_+d\sigma(\bx)
\nonumber
\\
&-\int_{\partial D}\overline{\mathcal{S}_D^\omega[\bpsi^\alpha_n]}\cdot\frac{\partial\bJ^\beta_m}{\partial\nu}d\sigma(\bx)
-
\int_{\partial D}\frac{\partial\overline{\bJ^\alpha_n}}{\partial\nu}\cdot\mathcal{S}^\omega_D[\bpsi^\beta_m]d\sigma(\bx)
 -\int_{\partial D}\frac{\partial\overline{\bJ^\alpha_n}}{\partial\nu}\cdot\bJ^\beta_m d\sigma(\bx).
\label{main5}
\end{align}
Similarly, we have
\begin{align}
W^{\beta,\alpha}_{n,m}=&
\int_{\partial D}\overline{\widetilde{\mathcal{S}}^\omega[\bphi^\beta_m]}\cdot\frac{\partial}{\partial\widetilde{\nu}}\widetilde{\mathcal{S}}^\omega_D[\bphi^\alpha_n]\Big|_-d\sigma(\bx)
-\int_{\partial D}\overline{\mathcal{S}^\omega_D[\bpsi^\beta_m]}\cdot\frac{\partial}{\partial\nu}\mathcal{S}^\omega_D[\bpsi^\alpha_n]\Big|_+d\sigma(\bx)
\nonumber
\\
&
-\int_{\partial D}\overline{\mathcal{S}_D^\omega[\bpsi^\beta_m]}\cdot\frac{\partial\bJ^\alpha_n}{\partial\nu}d\sigma(\bx)
-
\int_{\partial D}\frac{\partial\overline{\bJ^\beta_m}}{\partial\nu}\cdot\mathcal{S}^\omega_D[\bpsi^\alpha_n]d\sigma(\bx) -\int_{\partial D}\frac{\partial\overline{\bJ^\beta_m}}{\partial\nu}\cdot\bJ^\alpha_n d\sigma(\bx)
\nonumber
\\
=&
\int_{\partial D}\frac{\partial}{\partial\widetilde{\nu}}\overline{\widetilde{\mathcal{S}}^\omega[\bphi^\beta_m]}\Big|_-
\cdot\widetilde{\mathcal{S}}^\omega_D[\bphi^\alpha_n]d\sigma(\bx)
-\int_{\partial D}\frac{\partial}{\partial\nu}\overline{\mathcal{S}^\omega_D[\bpsi^\beta_m]}\Big|_+\cdot\mathcal{S}^\omega_D[\bpsi^\alpha_n]d\sigma(\bx)
\nonumber
\\
&
-\int_{\partial D}\overline{\mathcal{S}_D^\omega[\bpsi^\beta_m]}\cdot\frac{\partial\bJ^\alpha_n}{\partial\nu}d\sigma(\bx)
-
\int_{\partial D}\frac{\partial\overline{\bJ^\beta_m}}{\partial\nu}\cdot\mathcal{S}^\omega_D[\bpsi^\alpha_n]d\sigma(\bx)
 -\int_{\partial D}\overline{\bJ^\beta_m} \cdot\frac{\partial\bJ^\alpha_n}{\partial\nu}d\sigma(\bx).
\label{main6}
\end{align}
The proof is completed by taking complex conjugate of expression \eqref{main6} and comparing the result with equation \eqref{main5}.
\end{proof}

\section{Proof of Theorem \ref{thmOptic}}\label{AppendC}

In order to prove identity \eqref{W-constraint}, we follow the approach taken by \cite{Varath}. Since $\mathbf{W}_\infty$ is independent of the choice of incident field, we  consider the case when the plane waves
$$
\bu^{\rm inc}_P(\bx):= \nabla e^{i\K_P \bx\cdot\bd}
\quad\text{and}\quad
\bu^{\rm inc}_S(\bx):= \vec{\nabla }_\perp\times e^{i\K_S \bx\cdot\bd},
$$
are incident simultaneously and use the superposition principle for the optical theorem thanks to the linearity of the RHS of identity \eqref{ScatteringCrossSec}. Note that the coefficients
$a^\alpha_m (\bu^{\rm inc}_\alpha)$ and $\gamma^\alpha_n$ in this case are given by
\begin{align*}
a^\alpha_m(\bu^{\rm inc})= e^{im(\pi/2-\theta_\bd)}
\quad\text{and}\quad
\gamma^\alpha_n= \frac{i}{4\rho_0\omega^2}\sum_{m\in\ZZ}\left[ a^P_m W^{\alpha,P}_{m,n}+a^S_m W^{\alpha,S}_{m,n}\right].
\end{align*}
To facilitate ensuing discussion, let us define
$$
\mathbf{A}:=
\begin{pmatrix}
\mathbf{A}_P\\\mathbf{A}_S
\end{pmatrix}
\,\,\text{and}\,\,
\bgamma:=
\begin{pmatrix}
\bgamma_P\\\bgamma_S
\end{pmatrix},
\quad\text{with}\quad
\left(\mathbf{A}_\alpha\right)_m:= a^\alpha_m(\bu^{\rm inc})
\,\,\text{and}\,\,
\left(\bgamma_\alpha\right)_m:=\gamma^\alpha_m,
\quad\forall m\in\ZZ.
$$
It can be easily seen, by the definitions of $\mathbf{A}$ and $\bgamma$, and the fact that $\mathbf{W}_\infty$ is Hermitian,  that
\begin{align*}
\bgamma= \frac{i}{4\rho_0\omega^2} \mathbf{A}^\top \mathbf{W}_\infty
\quad\text{and}\quad
\bgamma\cdot\overline{\bgamma}= \frac{1}{(4\rho_0\omega^2)^2}\mathbf{A}^T\mathbf{W}_\infty\overline{\mathbf{W}_\infty}\overline{\mathbf{A}}.
\end{align*}
On the other hand, using the orthogonality relations \eqref{PP}-\eqref{SS} of the cylindrical surface vector potentials and fairly easy manipulations, we have
\begin{align}
\label{LHS}
\int_0^{2\pi} \left(\frac{1}{\K_P}\left|\bu^\infty_P(\hbx;\hbd)\right|^2+\frac{1}{\K_S}\left|\bu^\infty_S(\hbx;\hbd)\right|^2\right)d\theta
= 4 \bgamma \cdot \overline{\bgamma}= \frac{4}{(4\rho_0\omega^2)^2}\mathbf{A}^T\mathbf{W}_\infty\overline{\mathbf{W}_\infty}\overline{\mathbf{A}}.
\end{align}
Similarly, by virtue of  superposition principle, the RHS of the identity \eqref{ScatteringCrossSec}  can be written as
\begin{align}
2\Bigg[\sqrt{\frac{2\pi}{\K_P}}
\Im m \left\{\sqrt{i} \bu^\infty_P(\hbd;\hbd, P)\cdot\hat{\be}_r\right\} -
\sqrt{\frac{2\pi}{\K_S}}\Im m
& \left\{\sqrt{i} \bu^\infty_S(\hbd;\hbd, S)\cdot\hat{\be}_\theta\right\}\Bigg]
\nonumber
\\
&=\frac{4}{4\rho_0\omega^2}\Im m\left\{\mathbf{A}^\top\mathbf{W}_\infty \overline{\mathbf{A}}\right\}.\label{RHS}
\end{align}
Substituting \eqref{LHS} and \eqref{RHS} in \eqref{ScatteringCrossSec}, one gets
\begin{align}
\frac{1}{4\rho_0\omega^2} \mathbf{A}^T\mathbf{W}_\infty\overline{\mathbf{W}_\infty}\overline{\mathbf{A}}=\Im m\left\{\mathbf{A}^\top\mathbf{W}_\infty \overline{\mathbf{A}}\right\}.\label{first-relation}
\end{align}
Finally, note that
\begin{align*}
\Im m\left\{\mathbf{A}^\top\mathbf{W}_\infty \overline{\mathbf{A}}\right\}= & \Re e\left\{\mathbf{A}^\top\right\}\Im m\left\{\mathbf{W}_\infty \right\}\Re e\left\{\mathbf{A}\right\}-\Re e\left\{\mathbf{A}^\top\right\}\Re e\left\{\mathbf{W}_\infty \right\}\Im m\left\{\mathbf{A}\right\}
\\
&+\Im m\left\{\mathbf{A}^\top\right\}\Re e\left\{\mathbf{W}_\infty \right\}\Re e\left\{\mathbf{A}\right\}
+\Im m\left\{\mathbf{A}^\top\right\}\Im m\left\{\mathbf{W}_\infty\right\}\Im m\left\{\mathbf{A}\right\}.
\end{align*}
Recall that each term on the RHS of the above equation is a scalar and the matrix $\mathbf{W}_\infty$ is Hermitian. Thus, the second term  cancels out the third one on transposition. Finally, the first and the fourth terms can be combined to yield
\begin{align}
\Im m\left\{\mathbf{A}^\top\mathbf{W}_\infty \overline{\mathbf{A}}\right\}= \mathbf{A}^\top\Im m\left\{\mathbf{W}_\infty \right\}\overline{\mathbf{A}}.
\label{im-relation}
\end{align}
The relation \eqref{W-constraint} follows by substituting \eqref{im-relation} back in \eqref{first-relation}. This completes the proof.

\section{Proof of Theorem \ref{theorem-cloaking}}\label{AppendD}

Recall that, for $t\to 0$,
\begin{align*}
&J_n(t)=\phantom{-}\frac{t^n}{2^n\Gamma(n+1)}+O(t^{n+1}),
\\\nm
&J'_n(t)=\frac{nt^{n-1}}{2^n\Gamma(n+1)}+O(t^{n}),
\\\nm
& H_n^{(1)}(t)=-i\frac{2^n\Gamma(n)}{\pi t^n}+O(t^{-n+1}),
\\\nm
&(H_n^{(1)})'(t)=i\frac{2^n\Gamma(n+1)}{\pi t^{n+1}}+O(t^{-n}).
\end{align*}
Hence, by the definition of $B^\alpha_n(t,\lambda,\mu)$, $C^\alpha_n(t,\lambda,\mu)$, $\widehat{B}^\alpha_n(t,\lambda,\mu)$ and $\widehat{C}^\alpha_n(t,\lambda,\mu)$, we have
\begin{align*}
&B^P_n(t,\lambda,\mu)=-\frac{i\mu 2^{n+1}\Gamma(n+1)}{\pi t^n}+O(t^{-n+1}),
\\\nm
&C^S_n(t,\mu)=\frac{i\mu 2^{n+1}\Gamma(n+1)}{\pi t^n}+O(t^{-n+1}),
\\\nm
&C^P_n(t,\mu)=-B^S_n(t,\lambda,\mu)
=-\frac{\mu \,2^{n+1}\Gamma(n+1)\,(n+1)}{\pi t^n}+O(t^{-n+1}),
\\\nm
&\widehat{B}^P_n(t,\lambda,\mu)
=-\frac{\mu  t^n}{2^{n-1}\Gamma(n)}+O(t^{n+1}),
\\\nm
&\widehat{C}^S_n(t,\lambda,\mu)
=\frac{\mu  t^n}{2^{n-1}\Gamma(n)}+O(t^{n+1}),
\\\nm
&\widehat{C}^P_n(t,\lambda,\mu)=-\widehat{B}^S_n(t,\lambda,\mu)
=i\frac{\mu  (n-1) t^n}{2^{n-1}\Gamma(n)}+O(t^{n+1}),
\end{align*}
as $t\to 0$.
Inserting the previous asymptotic behavior into the expression of $\mathbf{M}_{n,j}$, we get
\begin{align*}
\mathbf{M}_{n,j}=
\begin{pmatrix}
\mathbf{A}_{11} & \mathbf{A}_{12}
\\
\mathbf{A}_{21} & \mathbf{A}_{22}
\end{pmatrix},
\end{align*}
where
\begin{align*}
\mathbf{A}_{11} &= \phantom{-} \frac{n}{2^n\Gamma(n+1)}
\begin{pmatrix}
t^{n}_{j,P} & i t^n_{j,S}
\\\nm
i t^{n}_{j,P} & -t^{n}_{j,S}
\end{pmatrix}
+O(\omega^{n+1}),
\\\nm
\mathbf{A}_{12} &=\phantom{-} \frac{2^n\Gamma(n+1)}{\pi}
\begin{pmatrix}
it^{-n}_{j,P} &  t^{-n}_{j,S}
\\\nm
t^{-n}_{j,P} & -it^{-n}_{j,S}
\end{pmatrix}
+O(\omega^{-n+1}),
\\\nm
\mathbf{A}_{21} &= -\frac{\mu }{2^{n-1}\Gamma(n)}
\begin{pmatrix}
-t^{n}_{j,P} &  -i(n-1) t^n_{j,S}
\\\nm
i(n-1) t^{n}_{j,P} &  t^{n}_{j,S}
\end{pmatrix}
+O(\omega^{n+1}),
\\\nm
\mathbf{A}_{22} &= -\frac{2^{n+1}\mu\Gamma(n+1)}{\pi}
\begin{pmatrix}
it^{-n}_{j,P} &  -(n+1)t^{-n}_{j,S}
\\\nm
(n+1)t^{-n}_{j,P} & -it^{-n}_{j,S}
\end{pmatrix}
+O(\omega^{-n+1}).
\end{align*}
It implies that
\begin{align}
\mathbf{M}_{n,j} &=
\begin{pmatrix}
O(\omega^{n+1}) & O(\omega^{-n+1})
\\\nm
O(\omega^{n}) & O(\omega^{-n})
\end{pmatrix},
\quad j=1,\cdots, L,\label{Mnj}
\\\nm
\mathbf{M}_{n,L} &=
\begin{pmatrix}
0 & 0
\\\nm
O(\omega^{n}) & O(\omega^{-n})
\end{pmatrix},
\quad \text{as}\;\; \omega\to 0.\label{MnL}
\end{align}
Moreover, the inverse of $\mathbf{M}_{n,j}$ can be expressed as
\begin{align*}
\mathbf{M}_{n,j}^{-1} &=
\begin{pmatrix}
\mathbf{A}_{11}^{-1}+\mathbf{A}_{11}^{-1}\mathbf{A}_{12}\mathbf{B}^{-1}  \mathbf{A}_{21}\mathbf{A}^{-1}_{11}
&
-\mathbf{A}_{11}^{-1}\mathbf{A}_{12} \mathbf{B}^{-1}
\\\nm
-{\mathbf{B}}^{-1}\mathbf{A}_{21}\mathbf{A}_{11}^{-1} & \mathbf{B}^{-1}
\end{pmatrix},
\end{align*}
where $\mathbf{B}$ is the Schur's complement of $\mathbf{A}_{22}$, that is,
$$
\mathbf{B}:=\mathbf{A}_{22}-\mathbf{A}_{21}\mathbf{A}_{11}^{-1}\mathbf{A}_{12}.
$$
Since
\begin{align*}
\mathbf{A}_{11}^{-1}=O(\omega^{-n-1}),
\quad
\mathbf{A}_{11}^{-1}\mathbf{A}_{12}=O(\omega^{-2n}),
\quad
\mathbf{A}_{21}\mathbf{A}_{11}^{-1}=O(\omega^{-1}),
\quad\text{and}\quad
\mathbf{B}^{-1}=O(\omega^n),
\end{align*}
it follows that
\begin{align}
\mathbf{M}^{-1}_{n,j}=
\begin{pmatrix}
O(\omega^{-n-1}) & O(\omega^{-n})
\\\nm
O(\omega^{n-1}) & O(\omega^{n})
\end{pmatrix},
\quad\text{as}\;\;\omega\to 0.\label{Mnjinv}
\end{align}
Inserting \eqref{Mnj}, \eqref{MnL} and \eqref{Mnjinv} into the expression  \eqref{eq:6} of $\mathbf{Q}^{(n)}$  and then making use of the series expansions of $J_n$, $Y_n$, $J'_n$ and $Y_n'$,  we find out that
\begin{align*}
\mathbf{Q}^{(n)}_{21}(\lambda,\mu,\rho\omega^2)
&=\omega^{n}
\left(
\bG_{n,0}(\lambda,\mu,\rho)+\sum_{l=1}^{N-n}\sum_{j=0}^{L+1} \bG_{n,l}^{(j)}(\lambda,\mu,\rho)\omega^{2l}(\ln\omega)^j + o\left(\omega^{2(N-n)}\right)
\right),
\\
\mathbf{Q}^{(n)}_{22}(\lambda,\mu,\rho\omega^2)
&=\omega^{-n}
\left(
\bH_{n,0}(\lambda,\mu,\rho)+\sum_{l=1}^{N-n}\sum_{j=0}^{L+1} \bH_{n,l}^{(j)}(\lambda,\mu,\rho)\omega^{2l}(\ln\omega)^j + o\left(\omega^{2(N-n)}\right)
\right),
\end{align*}
which together with \eqref{eq:7} yields \eqref{eq:8}. Here, the remaining terms $o(\omega^{2(N-n)})$ are understood element-wise for the matrices.

\bibliographystyle{plain}

\begin{thebibliography}{99}

\bibitem{Princeton}
Ammari, H.,  Bretin, E., Garnier, J., Kang, H. , Lee, H., Wahab, A.:  
Mathematical Methods in Elasticity Imaging. Princeton Series in Applied Mathematics. Princeton University Press, NJ (2015)  

\bibitem{iakovleva} 
Ammari, H., Calmon, P.,  Iakovleva, E.: 
Direct elastic imaging of a small inclusion. SIAM J. Imaging Sci. \textbf{1}(2), 169--187 (2008)

\bibitem{medium} 
 Ammari, H.,  Chow, Y. T.,  Zou, J.: The concept of heterogeneous scattering coefficients and its application in inverse medium scattering. SIAM J. Math. Anal. \textbf{46}(4), 2905--2935 (2014) 

\bibitem{MSRI-Book} 
 Ammari, H.,  Garnier, J.,  Jing, W.,  Kang, H.,  Lim, M.,  Solna, K.,   Wang, H.: Mathematical and Statistical Methods for Multistatic Imaging. Lecture Notes in Mathematics, vol. 2098. Springer-Verlag, Cham (2013)


\bibitem{Ammari1} 
Ammari, H.,  Kang, H.,  Lee, H.,   Lim, M.: Enhancement of near-cloaking using generalized polarization tensors vanishing structures. Part I: The conductivity problem. Comm. Math. Phys. \textbf{317}(1), 253--266 (2013)


\bibitem{Helmholtz} 
 Ammari, H.,  Kang, H.,  Lee, H.,  Lim, M.: Enhancement of near-cloaking. Part II: The Helmholtz equation. Commun. Math. Phys.  \textbf{317}(2), 485--502 (2013)

\bibitem{Maxwell} 
 Ammari, H.,  Kang, H.,  Lee, H.,  Lim, M.,  Yu, S.: Enhancement of near cloaking for the full Maxwell equations.   SIAM J. Appl. Math.  \textbf{73}(6), 2055--2076 (2013)

\bibitem{Shape} 
 Ammari, H.,  Tran, M. P.,  Wang, H.: Shape identification and classification in echolocation.  SIAM J. Imaging Sci.  \textbf{7}(3), 1883--1905 (2014)

\bibitem{BaoLiu}  Bao, G.,  Liu, H.: Nearly cloaking the electromagnetic fields. SIAM J. Appl. Math. \textbf{74}(3), 724--742 (2014)

\bibitem{BaoLiuZou}  Bao, G.,  Liu, H.,   Zou, J.: Nearly cloaking the full Maxwell equations: Cloaking active contents with general conducting layers. J. Math. Pures Appl. \textbf{101}(5), 716--733 (2014)

\bibitem{NullField}
 Bates, R. H. T.,  Wall, D. J. N.: Null field approach to scalar diffraction I. General method. Phil. Trans. R. Soc. A \textbf{287}(1339), 45--78 (1977)


\bibitem{Bergh} 
 Bergh,  J.,  L\"{o}fstr\"{o}m, J.: Interpolation Spaces. An Introduction. Grundlehren der Mathematischen Wissenschaften, vol. 223. Springer-Verlag, Berlin-New York (1976)

\bibitem{ChenChan}  Chen,  H.,   Chan, C.T.: Acoustic cloaking and transformation acoustics. J. Phys. D: Appl. Phys. \textbf{43}(11), 113001 (2010)

\bibitem{Dahlberg}
Dahlberg, B. E.,  Kenig, C. E.,  Verchota, G.: Boundary value problem for the systems of elastostatics in Lipschitz domains. Duke Math. Jour. \textbf{57}(3), 795--818 (1988) 

\bibitem{Dassios87}
 Dassios,  G.,  Kiriaki, K.: On the scattering amplitudes for elastic waves. Z. Angew. Math. Phys.  \textbf{38}(6), 856--873 (1987) 


\bibitem{Dassios2000}
 Dassios, G.,   Kleinman, R.: Low Frequency Scattering. Oxford University Press, Oxford (2000)


\bibitem{Diatta}  
Diatta,  A., Guenneau, S.: Controlling solid elastic waves with spherical cloaks. Appl. Phys. Lett.  
\textbf{105}, 021901 (2014)

\bibitem{Diatta2}
Diatta,  A.,   Guenneau, S.: Cloaking via change of variables in elastic impedance tomography. arXiv:1306.4647 (2013) 

\bibitem{Farhat} 
 Farhat, M.,  Guenneau, S.,  Enoch, S.,   Movchan, A.: Cloaking bending waves propagating in thin elastic plates. Phys. Rev. B  \textbf{79}, 033102 (2009) 

\bibitem{Ganesh10}
 Ganesh, M.,  Hawkins, S. C.: A far-field based T-matrix method for two dimensional obstacle scattering. ANZIAM J. \textbf{50}, C121--C136 (2010)


\bibitem{Ganesh}
 Ganesh, M.,   Hawkins, S. C.: Three dimensional electromagnetic scattering T-matrix computations.  J. Comput. Appl. Math.  \textbf{234}(6), 1702--1709 (2010) 


\bibitem{Greenleaf1} 
 Greenleaf, A., Kurylev, Y.,  Lassas, M.,   Uhlmann, G.: Cloaking devices, electromagnetic wormholes and transformation optics. SIAM Review \textbf{51}(1), 3--33 (2009) 

\bibitem{Greenleaf2} 
 Greenleaf, A.,  Kurylev, Y.,  Lassas, M.,   Uhlmann, G.: Invisibility and Inverse Problems. Bulletin AMS \textbf{46}(1), 55-97 (2009)

\bibitem{GLU} 
Greenleaf, A.,  Lassas, M.,  Uhlmann, G.: Anisotropic conductivities that cannot be detected by EIT.  Physiolog. Meas. \textbf{24}(2), 413--420 (2003)

\bibitem{GLU2} 
 Greenleaf, A.,  Lassas, M.,   Uhlmann,  G.: On nonuniqueness for Calder\'on's inverse problem. Math. Res. Lett.  \textbf{10}(5),  685--693 (2003) 

\bibitem{HL2015} 
 Hu, G.,   Liu, H.: Nearly cloaking the elastic wave fields. J. Math. Pures Appl. \textbf{104}(6), 
 1045--1074 (2015) 
 
\bibitem{Kupradze79}
Kupradze, V. D.,  Gegelia, T. G.,  Basheleishvili, M. O.,   Burchuladze, T. V.:  Three-dimensional Problems of the Mathematical Theory of Elasticity and Thermoelasticity.  North-Holland Publishing Company, Amsterdam-New York-Oxford (1979)


\bibitem{Leo} 
Leonhardt, U.:  Optical conformal mapping. Science \textbf{312}(5781), 1777--1780 (2006)  

\bibitem{LimYu} 
 Lim, M.,   Yu, S.: Reconstruction of the shape of an inclusion from elastic moment tensors.  In: Mathematical and Statistical Methods for Imaging, pp. 61--76, Contemp. Math., vol. 548. Amer. Math. Soc., Providence, RI (2011) 

\bibitem{Martin03}
 Martin, P. A.: On the connections between boundary integral equations and T-matrix methods. Eng. Anal. Bound. Elem. \textbf{27}(7), 771--777 (2003) 


\bibitem{Martin06}
 Martin, P. A.: 
 Multiple Scattering: Interaction of Time-harmonic Waves and N Obstacles. Cambridge University Press, New York (2006)


\bibitem{MBW} 
 Milton, G. W.,  Briane, M., and  Willis, J.R.: 
 On cloaking for elasticity and physical equations with a transformation invariant form.  New J. Phys. \textbf{8}, 248 (2006),


\bibitem{Rev1}
 Mishchenko, M. I.,  Videen, G.,  Babenko, V. A.,  Khlebtsov, N. G.,  Wriedt, T.: 
 T-matrix theory of electromagnetic scattering by particles and its applications: A comprehensive reference database.  J. Quant Spectrosc. Radiat. Transfer \textbf{88}(1-3), 357--406 (2004) 


\bibitem{Rev2}
Mishchenko, M. I.,  Videen, G.,  Babenko, V. A.,  Khlebtsov, N. G.,   Wriedt, T.: 
Comprehensive   T-matrix reference database: A 2004--2006 update. J. Quant Spectrosc. Radiat. Transfer  \textbf{106}(1-3), 304--324 (2007) 


\bibitem{Morse} 
 Morse, P. M.,  Feshbach, H.: 
 Methods of Theoretical Physics, vols. I and II.  McGraw-Hill, NY (1953)

\bibitem{nedelec}  N\'ed\'elec, J. C.:
 Acoustic and Electromagnetic Equations: Integral Representations for Harmonic Problems. App. Math. Sci.,  vol. 144. Springer-Verlag, New York (2001)

\bibitem{Norris11} 
 Norris,  A., Shuvalov,  A.: 
 Elastic cloaking theory. Wave Motion  \textbf{48}(6), 525--538  (2011)

\bibitem{NIST} 
Olver, F. W. J.,  Lozier, D. W.,  Boisvert, R. F.,   Clark, C. W. (eds.):  NIST Handbook of Mathematical Functions.  Cambridge University Press, New York (2010)

\bibitem{Ottaviani}
 Ottaviani, E.,   Pierotti, D., Reconstruction of scattering data by the optical theorem. In:   Proc. IEEE Ultrasonics Symp., 1989, IEEE, Piscataway, NJ, vol. 2, pp. 917--920 (1989)


\bibitem{Parnell} 
 Parnell,  W.: 
 Nonlinear pre-stress for cloaking from antiplane elastic waves.  Proc. Royal Soc. A  \textbf{468}(2138), 563--580 (2012) 

\bibitem{PenSchSmi} 
 Pendry, J.,  Schurig, D.,   Smith, D.: 
 Controlling electromagnetic fields.  Science  \textbf{312}(5781), 1780--1782 (2006) 

\bibitem{Sevroglou} 
 Sevroglou, V.,  Pelekanos, G.: 
 Two-dimensional elastic Herglotz functions and their applications in inverse scattering.  J.  Elast. \textbf{68}(1), 123--144 (2002) 

\bibitem{Rev3}
 Varadan, V. V.,   Lakhtakia, A.,  Varadan,  V. K. : 
 Comments on recent criticism of the T-matrix method. J. Acoust. Soc. Am. \textbf{84}(6), 2280--2284 (1988) 

\bibitem{Varadan80}
 Varadan, V. K.,  Varadan V. V.(eds.):  Electromagnetic and Elastic Wave Scattering - Focus on the T-matrix Approach.  Pergamon Press Inc.,  Oxford (1980)


\bibitem{Varath2}
 Varatharajulu, V.: 
 Reciprocity relations and forward amplitude theorems for elastic waves.  J. Math. Phys.  \textbf{18}(4), 537--543 (1977) 


\bibitem{Varath}  
Varatharajulu  V.,  Pao, Y. H.: 
Scattering matrix for elastic waves. I. Theory. J. acoust. Soc. Am. \textbf{60}(3), 556--566 (1976) 

\bibitem{Waterman-em}
 Waterman, P. C.: 
 Matrix formulation of electromagnetic scattering. Proc. IEEE \textbf{53}(8), 805--812 (1965) 


\bibitem{Waterman-a}
 Waterman, P. C.: 
 New formulation of acoustic scattering. J. Acoust. Soc. Am. \textbf{45}(6), 1417--1429 (1969) 


\bibitem{Waterman}
 Waterman, P. C.: 
 Matrix theory of elastic wave scattering.  J. Acoust. Soc. Am. \textbf{60}(3), 567--580 (1976) 

\end{thebibliography}

\end{document}